\documentclass[runningheads]{llncs}

\usepackage[T1]{fontenc}
\usepackage{graphicx}
\usepackage{hyperref}
\usepackage[usenames,dvipsnames]{xcolor}
\usepackage{color}

\usepackage{mathtools}
\usepackage{amsmath}
\usepackage{amssymb}
\usepackage{xspace}
\usepackage{proof}
\usepackage{stackengine}
\usepackage{xstring}
\usepackage{xparse}
\usepackage{cleveref}
\usepackage{stmaryrd}
\usepackage{subcaption}
\usepackage{tikz}
\usepackage{listings}
\usepackage{multirow}
\usepackage{xifthen}
\usepackage{bm}
\usepackage{booktabs}
\usepackage{cellspace}
\usepackage[inline]{enumitem}
\usepackage{algorithm2e}
\usepackage{environ}
\usepackage[normalem]{ulem}

\usetikzlibrary{arrows.meta,positioning,calc,shapes.geometric,patterns,patterns.meta,decorations.pathmorphing}
\usetikzlibrary{tikzmark,decorations.pathreplacing}

\newcommand{\Nats}{\mathbb{N}}
\newcommand{\emptyseq}{\varepsilon}
\newcommand{\filter}[1]{\vert_{#1}}

\newcommand{\CTLS}{$\textnormal{CTL}^{*}$\xspace}
\newcommand{\LTL}{\textnormal{LTL}\xspace}

\newcommand{\tand}{\textnormal{and}}
\newcommand{\tiff}{\textnormal{iff}}
\newcommand{\qiff}{\quad\tiff\quad}

\newcommand{\hAss}{\overline{\mathcal{K}}}

\newcommand{\cmDel}[2]{\textit{del}_{#1}(#2)}

\newcommand{\cmU}{\textit{UA}}

\tikzset{ver1/.style={draw,fill=white, minimum size=12pt,inner sep=0cm}}
\tikzset{ver/.style={draw,fill=white, minimum size=17pt}}
\tikzset{mod/.style={fill=white}}
\tikzset{inp/.style={circle, fill=black, text=white}}
\tikzset{col/.style={fill=#1!40}}
\tikzset{col/.default=white}

\makeatletter
\NewDocumentCommand{\overcolor}{mm}{\mathpalette\overundercolor@{\overline{#1}{#2}}}
\NewDocumentCommand{\undercolor}{mm}{\mathpalette\overundercolor@{\underline{#1}{#2}}}
\newcommand{\overundercolor@}[2]{\overundercolor@@#1#2}
\newcommand{\overundercolor@@}[4]{%
  \sbox\z@{$\m@th#1#3$}%
  \mathcolor{#4}{#2{\usebox\z@}}%
}
\makeatother

\newcommand{\textttn}[1]{\textup{\texttt{#1}}}

\newcommand{\hH}{\mathcal H}
\newcommand{\hJ}{\mathcal J}
\newcommand{\hK}{\mathcal K}
\newcommand{\hUU}{\mathcal U}
\newcommand{\hVV}{\mathcal V}

\newcommand{\hV}{V}

\newcommand{\hu}{u}
\newcommand{\hv}{v}
\newcommand{\hI}{I}
\newcommand{\hi}{i}
\newcommand{\hn}{n}

\newcommand{\hA}{{\rightarrow}}
\newcommand{\hE}{E}
\newcommand{\he}{e}
\newcommand{\hd}{\alpha}
\newcommand{\hC}{\gamma}
\newcommand{\hc}{c}
\newcommand{\hL}{\ell}
\newcommand{\hW}{W}
\newcommand{\hw}{w}
\newcommand{\hpath}{p}

\newcommand{\hHH}{\textttn{HG}}
\newcommand{\hHHV}{\textttn{HGV}}
\newcommand{\hCC}{\textttn{COL}}
\newcommand{\hLL}{\textttn{LAB}}
\newcommand{\hEE}{\textttn{ACT}}

\newcommand{\gA}{A}
\newcommand{\gB}{B}
\newcommand{\gC}{C}

\newcommand{\gN}{N}
\newcommand{\gZ}{Z}
\newcommand{\gS}{S}
\newcommand{\gP}{P}

\newcommand{\gR}{\mathcal R}

\newcommand{\gs}{\mathfrak S}
\newcommand{\gr}{\mathfrak T}

\newcommand{\gx}{\mathfrak C}
\newcommand{\gy}{\mathfrak D}

\newcommand{\gGG}{\mathcal G}
\newcommand{\gCC}{\mathcal C}

\newcommand{\grule}[3]{#3\Rightarrow#1(#2)} 

\newcommand{\lLL}{\mathcal L}

\newcommand{\tT}{T}
\newcommand{\tAss}{\fml{\tT}}

\newcommand{\bQ}{Q}
\newcommand{\bp}{p}
\newcommand{\bq}{q}
\newcommand{\br}{r}
\newcommand{\bF}{F}
\newcommand{\bS}{\Sigma}
\newcommand{\bs}{\sigma}
\newcommand{\bD}{\Delta}
\newcommand{\bI}{\bQ_0}
\newcommand{\bi}{\bq_0}

\newcommand{\bMM}{\mathcal M}
\newcommand{\bBB}[1][]{\textttn{BUCHI}\ifthenelse{\isempty{#1}}{}{(#1)}}

\newcommand{\HG}{Hypergraph\xspace}
\newcommand{\HGs}{Hypergraphs\xspace}
\newcommand{\HRG}{HRG\xspace}
\newcommand{\HRGs}{HRGs\xspace}

\newcommand{\bigO}{\mathcal O}

\newcommand{\tTree}{\textttn{Tree}\xspace}

\newcommand{\arity}{\textttn{ar}\xspace}
\newcommand{\first}{\textit{first}\xspace}
\newcommand{\last}{\textit{last}\xspace}

\newcommand{\cred}{{\color{red}\textttn{red}}\xspace}
\newcommand{\cblue}{{\color{blue}\textttn{blue}}\xspace}
\newcommand{\cgreen}{{\color{ForestGreen}\textttn{green}}\xspace}
\newcommand{\creds}{{\color{red}\textttn{r}}\xspace}
\newcommand{\cblues}{{\color{blue}\textttn{b}}\xspace}
\newcommand{\cgreens}{{\color{ForestGreen}\textttn{g}}\xspace}

\newcommand{\ConnPaths}{\textttn{Paths}\xspace}

\newcommand{\Traces}{\textttn{Traces}\xspace}
\newcommand{\FTraces}{\textttn{Traces}^{*}\xspace}
\newcommand{\ITraces}{\textttn{Traces}^{\omega}\xspace}

\newcommand{\Paths}{\Traces}
\newcommand{\Conn}{\ConnPaths}
\newcommand{\Runs}{\textttn{Runs}\xspace}
\newcommand{\CPaths}{\textttn{CPaths}\xspace}
\newcommand{\Run}{\seq{r}}
\newcommand{\Bool}{\textttn{BOOL}\xspace}
\newcommand{\Inf}{\textit{inf}\xspace}

\newcommand{\zip}{\mathbin{\scalebox{0.6}[1]{$\times\!\!\times$}}}
\renewcommand{\zip}{\times}

\newcommand{\bij}{\mathbin{\xrightarrow{\raisebox{-0.5ex}[0ex][0ex]{\tiny\text{bij}}}}}
\newcommand{\handle}[2][]{{#2}^{\ifthenelse{\isempty{#1}}{\circ}{#1}}}

\newcommand{\repl}[2]{#1[#2]}
\newcommand{\seq}[1]{\overline{\bm{#1}}}
\newcommand{\fml}[1]{\overline{#1}}

\newcommand{\boolT}{\textttn{T}}
\newcommand{\boolF}{\textttn{F}}
\newcommand{\boolb}{b}

\newcommand{\eqclass}[1]{\llbracket#1\rrbracket}
\newcommand{\sem}[1]{\llparenthesis#1\rrparenthesis}
\newcommand{\st}{\mathbin{s.t.}}
\newcommand{\fcong}{\mathrel{\vcenter{\offinterlineskip
  \halign{\hfil##\hfil\cr
  {\tiny\text{full}}\cr
  \noalign{\kern1pt}
  $\cong$\cr}}}}

\newcommand{\new}[1]{#1}

\newcommand{\repeattheorem}[1]{%
  	\begingroup
  	\renewcommand{\thetheorem}{\ref{#1}}%
  	\expandafter\expandafter\expandafter\theorem%
  	\csname reptheorem@#1\endcsname%
  	\endtheorem%
  	\endgroup%
}

\NewEnviron{reptheorem}[2][]{%
	\global\expandafter\xdef\csname reptheorem@#2\endcsname{%
    	\unexpanded\expandafter{\BODY}%
  	}
	\expandafter\ifthenelse{\isempty{#1}}{\theorem\label{#2}}{\theorem[#1]\label{#2}}\BODY\endtheorem
}

\newcommand{\repeatlemma}[1]{%
  	\begingroup%
  	\renewcommand{\thelemma}{\ref{#1}}%
  	\expandafter\expandafter\expandafter\lemma%
  	\csname replemma@#1\endcsname%
  	\endlemma%
  	\endgroup%
}

\NewEnviron{replemma}[2][]{%
  	\global\expandafter\xdef\csname replemma@#2\endcsname{%
  		\unexpanded\expandafter{\BODY}%
  	}%
	\expandafter\ifthenelse{\isempty{#1}}{\lemma}{\lemma[#1]}\label{#2}\BODY\endlemma
}

\begin{document}

\title{CTL* Model Checking on Infinite Families of Finite-State Labeled Transition Systems (Technical Report)}
\titlerunning{CTL* on Infinite Families}
\author{Roberto Pettinau\inst{1}\orcidID{0009-0003-1021-8708} \and Christoph Matheja\inst{1,2}\orcidID{000-0001-9151-0441}}
\authorrunning{R. Pettinau, C. Matheja}
\institute{Carl von Ossietzky University of Oldenburg, Oldenburg, Germany
\email{\{roberto.pettinau,christoph.matheja\}@uol.de}
\and 
DTU Compute, Kongens Lyngby, Denmark
}
\maketitle

\begin{abstract}
We study model checking algorithms for infinite families of finite-state labeled transition systems against temporal properties written in CTL*. Such families arise, for example, as models of highly configurable systems or software product lines.
 We model families using context-free graph grammars. We then develop a state labeling algorithm that works compositionally on the grammar's production rules with limited information about the context in which the rule is applied.
The result is a graph grammar modeling the same family but with extended labels.
We leverage this grammar to decide whether all, some, or (in)finitely many members of a family satisfy a given temporal property.
We have implemented our algorithms and present early experiments.
\end{abstract}

\section{Introduction}

Modern computer systems are increasingly built as highly-configurable product families that vary by features or performance parameters~(cf.~\cite{DBLP:journals/csur/ThumAKSS14,DBLP:books/daglib/0032924,DBLP:journals/tse/BeekLLV20}).
Such systems are inherently complex, because numerous feature combinations may result in product configurations that are deployed -- and potentially fail -- for the first time.
Hence, techniques to ensure the correctness of entire product families are crucial. 
One such technique is model checking, which has been successfully applied to certain classes of families, such as software product lines~\cite{DBLP:conf/icse/ClassenHSL11,DBLP:journals/tse/BeekLLV20,DBLP:journals/fac/ChrszonDKB18}.
Existing model-checking approaches often reason about families of structurally similar systems that are derived from a finite-state super-system, where all features are enabled~\cite{DBLP:journals/fac/ChrszonDKB18}.
However, complex product families may not admit such an all-in-one super-system.
For example, if the number of some components is configurable, one needs to consider infinite families of structurally evolving systems. 
Similar families arise in distributed systems with configurable network topologies~\cite{DBLP:journals/sigact/FoersterS19}.

In this paper, we develop model checking algorithms for infinite families $\{T_i\}_{i \in \Nats}$ of labeled transition systems (LTSs) that can evolve in size and structure. While every LTS $T_i$ is finite-state, it can be arbitrarily large.
To represent such families, we use hyperedge replacement grammars (HRGs), a well-studied extension of context-free grammars to describe sets of graphs~\cite{DBLP:conf/gg/DrewesKH97,engelfriet1997context}.
\Cref{fig:grammar-dll} depicts an HRG modeling a family of colored doubly-linked lists (see \Cref{fig:intro-2}), which will serve us as a running toy example.
The boxes are hyperedges. They correspond to nonterminals that are replaced by graphs according to the grammar's rules, see \Cref{fig:hyperedge-replacement} for an illustration.
Node colors indicate atomic propositions.
Our algorithms are based on a solution to the \emph{recoloring problem}, which informally reads as follows:
Given a family $F$ modeled by an HRG $\gGG$ and an $\omega$-regular property $\varphi$, 
compute a \emph{recolored} HRG $\gGG_{\varphi}$ modeling the same family $F$ with one exception: For every $T_i \in F$ and every state $s$ of $T_i$ with $(T_i, s) \models \varphi$, the state $s$ is additionally colored with $\varphi$.

Being able to compute recolored HRGs $\gGG_{\varphi}$ recursively yields decision procedures for various verification questions, including:
\begin{itemize}
	\item Does every (or some) LTS in the family $F$ satisfy the $\omega$-regular property $\varphi$?
	\item Are there finitely (or infinitely) many LTSs in $F$ that satisfy $\varphi$?
	\item Given formula $\Psi$ in \CTLS~\cite{DBLP:journals/jacm/EmersonH86}, does every (or some) LTS in $F$ satisfy $\Psi$?
	\item Given a formula $\psi$ in the qualitative fragment of the probabilistic branching time logic PCTL~\cite{DBLP:journals/fac/HanssonJ94}\cite[p. 787]{baierKatoen2008}, does every \emph{Markov chain} in $F$ satisfy $\psi$?
\end{itemize}
To illustrate how we compute recolored HRGs, consider the LTS in \Cref{fig:intro-2}, which is a member of the family specified by our example grammar in \Cref{fig:grammar-dll}.
As indicated by the rectangles, the LTS can be decomposed into three parts: (i) an application of the production rule $\gR_2$ (green box), (ii) the graph $\hK$ plugged into the nonterminal of $\gR_2$, and (iii) the context $\hJ$ in which the rule has been applied.
To determine whether the node $v$ of $\gR_2$ should be colored with an $\omega$-regular property, say $\cred~U~\cblue$, we do not need full knowledge about $\hJ$ and $\hK$.
It suffices to know that (1) we can reach a \cblue node from $v$ by taking an edge from $\gR_2$ into $\hJ$ via $\hK$ and that (2) we can globally stay in \cred nodes by taking edges from $\gR_2$ into $\hJ$ or $\hK$.
Due to (2), $\cred~U~\cblue$ is not satisfied for all traces starting in $v$. Hence $v$ should \emph{not} be colored with $\cred~U~\cblue$.
More generally, we derive an equivalence relation over graphs from any $\omega$-regular property that groups the possible graphs $\hJ$ and $\hK$ composed with the rule $\gR_2$ into classes that represent observations like (1) and (2).
We then show that those equivalence relations form a congruence and that it suffices to consider only finitely many equivalence classes.
Hence, we can solve the recoloring problem through a compositional analysis of the grammar's production rules.
The same approach can be applied recursively to reason about \CTLS instead of $\omega$-regular properties.

\paragraph{Contributions.} Our main contributions can be summarized as follows:
\begin{itemize}
	\item We develop an algorithm for the novel \emph{recoloring problem} for infinite families of finite-state labeled transition systems modeled by graph grammars.
	\item We prove our algorithm for the recoloring problem sound and complete.
	\item We extend our solution to the recoloring problem to decide whether all, some, or (in)finitely many members of an LTS family satisfy a CTL* formula.
	\item We have implemented our algorithms. Experiments demonstrate that properties in CTL* and in the qualitative fragment of PCTL can be verified within seconds for simple families of finite-state systems.
\end{itemize}

\paragraph{Outline.}
\Cref{sec:preliminaries} introduces preliminaries and formally states the recoloring problem. 
 \Cref{sec:refinement} presents a general approach for compositional reasoning about HRGs that we apply in \Cref{sec:meq} to the recoloring problem.
 \Cref{sec:impl} develops the resulting algorithm for recoloring and extends it to \CTLS.
 \Cref{sec:eval} reports on a prototypical implementation and experiments.
 Finally, \Cref{sec:related} discusses related work and concludes.
 Proofs and further details about our experiments are available online in a technical report~\cite{techRep}.

\begin{figure}[t]
\centering
\begin{subfigure}[b]{0.26\linewidth}
\centering
\scalebox{0.65}{\begin{tikzpicture}[thick,shorten >=8pt,shorten <=8pt]
	\def\xa{0}\def\ya{0}
    \draw (\xa+0,\ya) -- (\xa+1,\ya) node[pos=0.5, above, yshift=-1.5pt] {$\scriptstyle \mathbf 1$};
    \draw (\xa+1,\ya) -- (\xa+2,\ya) node[pos=0.5, above, yshift=-1.5pt] {$\scriptstyle \mathbf 2$};
    \draw[->] (\xa+2,\ya) to[out=-60,in=-120,distance=15] (\xa+3,\ya);
    \draw[->] (\xa+3,\ya) to[out=+120,in=+60,distance=15] (\xa+2,\ya);
    \draw (\xa+0,\ya) node[ver,circle,col=red] {};
    \draw (\xa+1,\ya) node[ver,mod] {$\gA$};
    \draw (\xa+2,\ya) node[ver,circle,col=red] {};
    \draw (\xa+3,\ya) node[ver,circle,col=blue] {};
    \draw (\xa-1.5,\ya) node {$\mathbf S := $};
    \draw (\xa-.75,\ya) node {$\gR_3$:};
    \def\xb{0}\def\yb{-1}
    \draw[->] (\xb+0,\yb) to[out=-60,in=-120,distance=15] (\xb+1,\yb);
    \draw[->] (\xb+1,\yb) to[out=+120,in=+60,distance=15] (\xb+0,\yb);
    \draw     (\xb+1,\yb) -- (\xb+2,\yb) node[pos=0.5, above, yshift=-1.5pt] {$\scriptstyle \mathbf 1$};
    \draw     (\xb+2,\yb) -- (\xb+3,\yb) node[pos=0.5, above, yshift=-1.5pt] {$\scriptstyle \mathbf 2$};
    \draw (\xb+0,\yb) node[ver,inp] {$\mathbf{1}$};
    \draw (\xb+1,\yb) node[ver,circle,col=red] {};
    \draw (\xb+2,\yb) node[ver,mod] {$\gA$};
    \draw (\xb+3,\yb) node[ver,inp] {$\mathbf{2}$};
    \draw (\xb-1.5,\yb) node {$\mathbf A := $};
    \draw (\xb-.75,\yb) node {$\gR_2$:};
    \def\xc{0}\def\yc{-2}
    \draw[->] (\xc+0,\yc) to[out=-60,in=-120,distance=15] (\xc+1,\yc);
    \draw[->] (\xc+1,\yc) to[out=+120,in=+60,distance=15] (\xc+0,\yc);
    \draw (\xc+0,\yc) node[ver,inp] {$\mathbf{1}$};
    \draw (\xc+1,\yc) node[ver,inp] {$\mathbf{2}$};
    \draw (\xc-1.3,\yc+.5) -- (\xc-1.3,\yc-.5);
    \draw (\xc-.75,\yc) node {$\gR_1$:};
    \draw (4,0) node {};
\end{tikzpicture}}
\caption{An HRG.}
\label{fig:grammar-dll}
\end{subfigure}	
\hfill
\begin{subfigure}[b]{0.39\linewidth}
\centering
\scalebox{0.65}{\begin{tikzpicture}[thick,shorten >=8pt,baseline=(current bounding box.center)]
    \draw[dashed] (0,0) -- (0,-1);
    \draw[dashed] (3,0) -- (3,-1);

    \draw[->] (3,0) to[out=-60,in=-120,distance=15] (4,0);
    \draw[->] (4,0) to[out=+120,in=+60,distance=15] (3,0);
    \draw (0,0) -- (1.5,0) node[pos=0.5, above, yshift=-1.5pt] {$\scriptstyle \mathbf 1$};
    \draw (1.5,0) -- (3,0) node[pos=0.5, above, yshift=-1.5pt] {$\scriptstyle \mathbf 2$};
    \draw (0,0) node[ver,circle,col=red] {};
    \draw (1.5,0) node[ver,mod] {$A$};
    \draw (3,0) node[ver,circle,col=red] {};
    \draw (4,0) node[ver,circle,col=blue] {};
   	
    \draw[->] (0,-1) to[out=-60,in=-120,distance=15] (1,-1);
    \draw (1,-1) -- (2,-1) node[pos=0.5, above, yshift=-1.5pt] {$\scriptstyle \mathbf 1$};
    \draw (2,-1) -- (3,-1) node[pos=0.5, above, yshift=-1.5pt] {$\scriptstyle \mathbf 2$};
    \draw[->] (1,-1) to[out=+120,in=+60,distance=15] (0,-1);
    \draw (0,-1) node[ver,inp] {$\mathbf{1}$};
    \draw (1,-1) node[ver,circle,col=red] {};
    \draw (2,-1) node[ver,mod] {$A$};
    \draw (3,-1) node[ver,inp] {$\mathbf{2}$};
    
    \draw[->] (0,-2) to[out=-60,in=-120,distance=15] (1,-2);
    \draw (1,-2) -- (2,-2) node[pos=0.5, above, yshift=-1.5pt] {$\scriptstyle \mathbf 1$};
    \draw (2,-2) -- (3,-2) node[pos=0.5, above, yshift=-1.5pt] {$\scriptstyle \mathbf 2$};
    \draw[->] (1,-2) to[out=+120,in=+60,distance=15] (0,-2);
    \draw[->] (3,-2) to[out=-60,in=-120,distance=15] (4,-2);
    \draw[->] (4,-2) to[out=+120,in=+60,distance=15] (3,-2);
    \draw (0,-2) node[ver,circle,col=red] {};
    \draw (1,-2) node[ver,circle,col=red] {};
    \draw (2,-2) node[ver,mod] {$A$};
    \draw (3,-2) node[ver,circle,col=red] {};
    \draw (4,-2) node[ver,circle,col=blue] {};
    
    \draw (-.25,-0) node[left] {$\gR_3$};
    \draw (-.25,-1) node[left] {$\gR_2$};
    \draw (-.25,-2) node[left] {$\repl{\gR_3}{\gR_2}$};
\end{tikzpicture}}
\caption{Hyperedge replacement.}
\label{fig:hyperedge-replacement}
\end{subfigure}
\hfill
\begin{subfigure}[b]{0.32\linewidth}
\centering
\begingroup\def\sx{1}
\scalebox{0.65}{\begin{tikzpicture}[thick]
    \draw[thick, line width=1pt, fill=magenta!5] (-0.5*\sx,1.75) rectangle (5.5*\sx,-1.1);
    \draw[thick, line width=1pt, fill=ForestGreen!5] (0.9*\sx,1.5) rectangle (4.1*\sx,-0.9);
    \draw[thick, line width=1pt, fill=cyan!5] (2.1*\sx,1.25) rectangle (3.9*\sx,-0.7);
    
    \node (V1) at (0*\sx,0) [ver,circle,col=red] {};
    \node (V2) at (1*\sx,0) [ver,circle,col=red] {};
    \node (V3) at (2*\sx,0) [ver,circle,col=red] {$\pmb \hv$};
    \node (V4) at (3*\sx,0) [ver,circle,col=red] {};
    \node (V5) at (4*\sx,0) [ver,circle,col=red] {};
    \node (V6) at (5*\sx,0) [ver,circle,col=blue] {};
    \draw[->] (V1) to[out=-60,in=-120,distance=10] (V2);
    \draw[->] (V2) to[out=+120,in=+60,distance=10] (V1);
    \draw[->] (V2) to[out=-60,in=-120,distance=10] (V3);
    \draw[->] (V3) to[out=+120,in=+60,distance=10] (V2);
    \draw[->] (V3) to[out=-60,in=-120,distance=10] (V4);
    \draw[->] (V4) to[out=+120,in=+60,distance=10] (V3);
    \draw[->] (V4) to[out=-60,in=-120,distance=10] (V5);
    \draw[->] (V5) to[out=+120,in=+60,distance=10] (V4);
    \draw[->] (V5) to[out=-60,in=-120,distance=10] (V6);
    \draw[->] (V6) to[out=+120,in=+60,distance=10] (V5);

    \draw (-0.15*\sx,1.5) node {$\hJ$};
    \draw (1.25*\sx,1.25) node {$\gR_2$};
    \draw (2.45*\sx,1.0) node {$\hK$};
\end{tikzpicture}} 	
\caption{Decomposition of LTS.}
\label{fig:intro-2}
\endgroup	
\end{subfigure}
\caption{Toy example of a family of LTSs modeled by an HRG. Each LTS in the family is a doubly linked list where all nodes but the rightmost one are red.}
\label{fig:overview}
\end{figure}
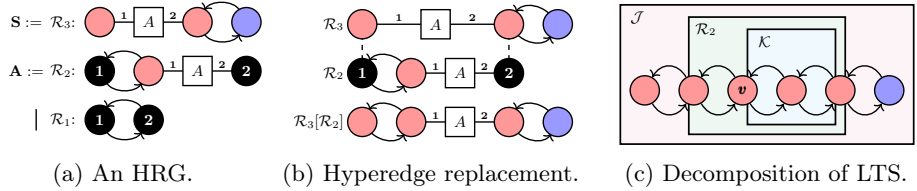

\section{Preliminaries and Formal Problem Statement}\label{sec:preliminaries}
Our model checking technique for families is based on three ingredients that we recap first:
B\"uchi automata, hypergraphs, and hyperedge replacement grammars.

\smallskip
\noindent
\emph{Basic notation.}
$X^{*}$ ($X^{\omega}$) is the set of all finite (infinite) sequences $\seq{x} = x_1 x_2 x_3 \ldots$ over the set $X$.
The (possibly infinite) length of $\seq{x}$ is $|\seq{x}|$. 
The empty sequence is $\emptyseq$.
We denote by $\seq{x}\filter{Y}$ the sequence obtained by keeping in $\seq{x}$ only elements in the set $Y$.
The concatenation of $\seq{x}$ and $\seq{y}$ is $\seq{x}\seq{y}$.
We denote by $X \uplus Y$ the union of the sets $X$ and $Y$ if $X \cap Y = \emptyset$; otherwise, $X \uplus Y$ is undefined.
For functions $f\colon X \to A$ and $g\colon Y \to B$, the function
$f \uplus g\colon (X \uplus Y) \to (A \cup B)$ maps every $x \in X$ to $f(x)$ and every $y \in Y$ to $g(y)$.
We lift functions to sequences and sets by pointwise application, e.g. $f(x_1 x_2 \ldots) = f(x_1) f(x_2) \ldots$ and $f(Y) = \{ f(y) ~|~ y \in Y \}$.


\subsection{B\"uchi Automata}\label{sec:props}
\CTLS model checking (cf.~\cite[p. 429]{baierKatoen2008}) based on recursively labeling states with satisfaction sets turns any \LTL model checking algorithm into a \CTLS model checker.
We thus first study model checking of infinite families for $\omega$-regular properties specified by B\"uchi automata, which subsumes \LTL (cf.~\cite{thomas1990automata,baierKatoen2008}).
\begin{definition} \label{def:buchi}
A \emph{B\"uchi automaton} is a tuple $\bMM = (\bQ, \bS, \bD, \bI, \bF)$, where $\bQ$ is a finite \emph{set of states}, $\bS$ is a finite \emph{alphabet}, $\bD \subseteq \bQ \times \bS \times \bQ$ is a \emph{transition relation}, $\bI \subseteq \bQ$ is a set of \emph{initial states}, and $\bF \subseteq \bQ$ is a set of \emph{final states}. 
We denote by $\bBB[\bS]$ the \emph{set of all B\"uchi automata} over the alphabet $\bS$.
\end{definition}
A B\"uchi automaton $\bMM = (\bQ, \bS, \bD, \bI, \bF)$ \emph{accepts} an infinite word $\seq\bs \in \bS^{\omega}$ if it admits a run that visits a final state infinitely often.
We define both finite and infinite runs as our model checking algorithm uses them for abstraction.
A \emph{finite} (resp. \emph{infinite}) \emph{run} of $\bMM$ is a sequence $\Run = \bq_0\sigma_0\cdots\sigma_{n-1}\bq_n$ (resp. $\Run =\bq_0\sigma_0\bq_1\sigma_1\cdots$) such that $(\bq_i,\sigma_i,\bq_{i+1}) \in \bD$ for all $0 \leq i < n$ (resp. $i \geq 0$).
We denote by $\Runs^*$ and $\Runs^\omega$ the \emph{sets of all finite and infinite runs} of $\bMM$, respectively.
$\Inf(\Run)$ denotes the set of states that appear \emph{infinitely often} in run $\Run$.
We extend the transition relation $\bD$ to take finite (infinite) words and indicate with a flag $\boolb \in \{\boolF,\boolT\}$ whether some final states have been visited (infinitely often): 
\begin{align*}
\bD^{*} ~=~ & 
\{ (\bq_0, \Run \vert_{\bS}, \bq_n, \boolb) \mid  \Run \in \Runs^{*}, \bq_0 = \first(\Run), \bq_n = \last(\Run), \boolb = (\Run \vert_{\bF} \ne \emptyseq)  \}
\\
\bD^{\omega} ~=~ &
\{(\bq_0, \Run \vert_{\bS}, \boolb) \mid \Run \in \Runs^{\omega}, \bq_0 = \first(\Run), b = (\Inf(\Run) \cap \bF \ne \emptyset)\}
\end{align*} 
The language of $\bMM$ is defined as
$
	\lLL(\bMM) = \{\seq\bs \in \bS^{\omega} \mid \exists \bi \in \bI.\; (\bi,\seq\bs,\boolT) \in \bD^\omega\}
$. 


\subsection{Hypergraphs} \label{sec:hypergraphs}

Hypergraphs (HGs) generalize ordinary graphs by allowing \emph{hyperedges} to connect an arbitrary number of nodes.  
To simplify constructions, we distinguish between ordinary edges and hyperedges.
From now on, we fix a set $\hCC$ of node \emph{colors} (think: atomic propositions), a set $\hEE$ of \emph{actions} (think: ordinary edge labels), and a set $\hLL$ of \emph{hyperedge labels}. 
Moreover, we equip HGs with a fixed number of uncolored abstract nodes that will be used for gluing HGs together. 

\begin{definition} \label{def:hypergraph}
A \emph{hypergraph} is a tuple $\hH = (\hV,\hn,\hA,\hE,\hd,\hC,\hL)$ whose set of \emph{nodes} $\hW = \hV \uplus \hI$ is partitioned into a finite set $\hV$ of \emph{concrete nodes} and a set $I = \{1,2,\ldots,\hn\}$ of \emph{abstract nodes} given by $\hn \in \Nats$.
Moreover, $\hA \subseteq \hW \times \hEE \times \hW$ is the labeled \emph{edge relation}, $\hE$ is a finite set of \emph{hyperedges}, and $\hd\colon \hE \to \hW^{*}$ assigns a sequence of \emph{attached nodes} to every hyperedge.
Finally, $\hC\colon \hV \rightarrow 2^{\hCC}$ \emph{assigns a set of colors to every concrete node} and $\hL\colon \hE \to \hLL$ \emph{assigns labels to hyperedges}.
We denote by $\hHH$ the set of all hypergraphs.
\end{definition}

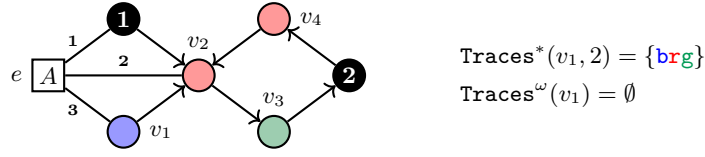
\begin{figure}[t]
\centering
\begingroup
\def\sx{1}
\def\sy{0.75}
\begin{tikzpicture}[thick]
    \node (I1) at (0*\sx,2*\sy) [ver1,inp] {$\mathbf 1$};
    \node (I2) at (3*\sx,1*\sy) [ver1,inp] {$\mathbf 2$};
    
    \node (V1) at (0*\sx,0*\sy) [ver1,circle,col=blue,label=0:$\hv_1$] {};
    \node (V2) at (1*\sx,1*\sy) [ver1,circle,col=red,label=$\hv_2$] {};
    \node (V3) at (2*\sx,0*\sy) [ver1,circle,col=ForestGreen,label=$\hv_3$] {};
    \node (V4) at (2*\sx,2*\sy) [ver1,circle,col=red,label=0:$\hv_4$] {};

    \node (E1) at (-1*\sx,1*\sy)[ver1,mod,label=180:$\he$] {$\gA$};

    \draw[->] (I1) -- (V2);
    \draw[->] (V1) -- (V2);
    \draw[->] (V2) -- (V3);
    \draw[->] (V3) -- (I2);
    \draw[->] (I2) -- (V4);
    \draw[->] (V4) -- (V2);
    \draw     (E1) -- (I1) node [pos=0.5, left , yshift=2pt] {$\scriptstyle \mathbf 1$};
    \draw     (E1) -- (V2) node [pos=0.5, above            ] {$\scriptstyle \mathbf 2$};
    \draw     (E1) -- (V1) node [pos=0.5, left, yshift=-2pt] {$\scriptstyle \mathbf 3$};
    
    \node[right,xshift=10pt] at (4*\sx,1.33*\sy) 
    	{$\Traces^*(\hv_1,2) = \{{\color{blue}{\texttt b}}{\color{red}{\texttt r}}{\color{ForestGreen}{\texttt g}}\}$};
    \node[right,xshift=10pt] at (4*\sx,0.66*\sy) 
    	{$\Traces^\omega(\hv_1) = \emptyset$};
\end{tikzpicture}
\endgroup
\caption{Illustration of a hypergraph together with two sets of traces.}
\label{fig:paths}
\end{figure}
We write $u \xrightarrow{a} v$ instead of $(u,a,v) \in \hA_{\hH}$.
Moreover, we write $\hV_\hH$, $\hn_\hH$, $\hW_{\hH}$, and so on, to refer to the components of a hypergraph $\hH$.
\Cref{fig:paths} (left) illustrates an HG with four concrete nodes $\hv_1-\hv_4$ with colors
\cblue, \cred, or \cgreen, two abstract nodes $1$ and $2$, one hyperedge $\he$ with label $\gA$, and six ordinary edges indicated by arrows. For simplicity, we omit ordinary edge labels.
The numbered connections to $\he$ indicate the sequence of attached nodes, i.e. $\hd(e) = 1 v_2 v_3$.

Hyperedges act as placeholders for HGs that can be inserted into the current hypergraph, while abstract nodes serve as connection points through which the current hypergraph can be embedded into other HGs.
An HG with neither abstract nodes nor hyperedges is a labeled transition system (LTS)~\cite[p. 20]{baierKatoen2008}.\footnote{The described LTS has no initial states, but one can use a color to model them.}
We reason about $\omega$-regular properties of such LTSs by considering their traces, i.e. sequences of node colors. To formalize traces, we first define the set $\ConnPaths_{\hH}$ of all \emph{paths} of $\hH \in \hHH$ using ordinary edges:
\[
\ConnPaths_{\hH} ~=~ \left\{ \seq{\hpath} \in \hW_{\hH}^{*} \uplus \hW_{\hH}^{\omega} \middle\bracevert \forall i \in [0,|\seq{\hpath}|).~ \exists a \in \hEE.~ \hpath_i \xrightarrow{a}_{\hH} \hpath_{i+1} \right\}
\]
A trace is then obtained from a path by considering the colors of each node.
Formally, the set $\FTraces_{\hH}(u,v)$ of \emph{finite traces} in $\hH$ from $u \in \hW_{\hH}$ to $v \in \hW_{\hH}$ is
\[
\FTraces_{\hH}(\hu,\hv) ~=~ \left\{ \hC_{\hH}\left(\hu\seq\hpath\hv \vert_{\hV_{\hH}}\right) \middle\bracevert \hu\seq\hpath\hv \in \Conn_{\hH} ~\tand~ \seq\hpath \in \hV_\hH^* \right\}.
\]
Notice that traces do not permit abstract nodes in the middle of a path.
The set $\ITraces_{\hH}(u)$ of \emph{infinite traces} in $\hH$ starting in $u \in \hW_{\hH}$ is 
\[
\ITraces_{\hH}(\hu) ~=~ \left\{ \hC_{\hH}\left(\hu\seq\hpath \vert_{\hV_{\hH}}\right) \middle\bracevert \hu\seq\hpath \in \Conn_{\hH} ~\tand~ \seq\hpath \in \hV_\hH^\omega \right\}.
\]
\Cref{fig:paths} (right) depicts examples of sets of traces for the HG on the left.
The set $\ITraces(v_1)$ is empty, because every infinite path starting in $\hv_1$ has to pass through the abstract node $2$.

As for transition systems, an HG $\hH$ and a node $\hv \in \hW_{\hH}$ \emph{satisfy} a property given by a B\"uchi automaton $\bMM \in \bBB[2^\hCC]$ iff all infinite traces starting in $v$ are accepted by $\bMM$, i.e.
$\hH, \hv \models \bMM ~\tiff~ \ITraces_{\hH}(\hv) \subseteq \lLL(\bMM)$.
\subsection{Hyperedge Replacement Grammars} \label{sec:hrgs}
To describe families of transition systems, we define how hyperedges are replaced by HGs.
A (hyperedge) \emph{assignment} is a function $\hAss\colon \hE_{\hH} \to \hHH$ that maps every hyperedge $\he \in \hE_{\hH}$ with $k$ attached nodes to an HG with $k$ abstract nodes, i.e. $|\hd_{\hH}(\he)| = \hn_{\hAss(\he)}$.
Hyperedge replacement then substitutes every hyperedge $e$ by the graph $\hAss(\he)$ while identifying the $i$-th attached node of $\he$ with the $i$-th abstract node of $\hAss(\he)$.
Formally, the HG $\repl{\hH}{\hAss} = (\hV,\hn_{\hH},\hA,\hE,\hd,\hC,\hL)$ obtained from \emph{hyperedge replacement} according to assignment $\hAss$ is given by:
\begin{align*}
\hV = \hV_{\hH} \uplus \biguplus_{\he \in \hE_\hH} \hV_{\hAss(\he)}
\quad
\hE = \biguplus_{\he \in \hE_{\hH}} \hE_{\hAss(\he)}
\quad
\hC = \hC_{\hH} \uplus \biguplus_{\he \in \hE_\hH} \hC_{\hAss(\he)}
\quad
\hL = \biguplus_{\he \in \hE_\hH} \hL_{\hAss(\he)}
\end{align*}
Moreover, for every $\he \in \hE_\hH$, let $f_{\he}(i)$ identify the $i$-th abstract node of $\hAss(e)$ with the $i$-th node attached to $\he$, i.e. $f_{\he}(i) = {\hd_{\hAss(e)}}_{i}$ if $i \in \hI_{\hAss(e)}$ and $f_{\he}(v) = v$ if $v \in \hV_{\hAss(e)}$.
Then the edge relation $\hA$ and assignment of attached nodes $\hd$ are 
\begin{align*}
\hA = \hA_{\hH} \uplus \biguplus_{\he \in \hE_\hH} \{ (f_{\he}(u), a, f_{\he}(v)) ~|~ u \xrightarrow{a}_{\hAss(e)} v \} 
\quad\tand\quad 
\hd = \biguplus_{\he \in \hE_{\hH}} f_{\he} \circ \hd_{\hAss(\he)}.
\end{align*}
If an HG $\hH$ has only a single hyperedge, which we replace by $\hK$, we write $\repl{\hH}{\hK}$ instead of $\repl{\hH}{\lambda \he. \hK}$.
Throughout this paper, we consider HGs up to isomorphism. Hence, nodes and hyperedges can always be renamed before replacement.

A trivial example of hyperedge replacement is the identity.
The corresponding assignment maps every hyperedge with label $\gA$ and $k$ attached nodes to the $k$-\emph{handle} $\handle{\gA} = (\emptyset,k,\emptyset,\{\he\}, \lambda \he.\, (1\,\ldots\,k), \lambda \he.\, \gA)$, consisting of one hyperedge labeled $\gA$ and $k$ abstract nodes attached to it in ascending order.\footnote{We use lambda expressions to denote functions. For example, $\lambda \he.\, \gA$ is the function that maps every hyperedge $\he$ to the nonterminal $\gA$.}

We model families of LTSs with hyperedge replacement grammars (\HRGs).

\begin{definition}
An \emph{\HRG} is a tuple $\gGG = (\gN,\gZ,\gP)$, where $\gN$ is a finite set of \emph{nonterminals} equipped with an arity $\arity\colon \gN \to \Nats$, $\gZ \subseteq \gN$ is a set of \emph{start symbols}, and $\gP \subseteq \gN \times \hHH$ is a set of \emph{production rules} such that, for all $(\gA,\hH) \in \gP$, (i) $\arity(\gA) = \hn_{\hH}$ and (ii) for all $\he \in \hE_{\hH}$, $\hL(\he) \in \gN$ and $|\hd_\hH(\he)| = \arity(\hL(\he))$.
\end{definition}
Conditions (i) and (ii) ensure for every production rule $(\gA,\hH) \in \gP$ that replacing a hyperedge labelled with $\gA$ by $\hH$ is well-defined. We will define the graph language of an \HRG in terms of its derivation trees, which we introduce first.

\begin{definition}
Let $\gGG = (\gN, \gZ, \gP)$ be an \HRG and $\gA \in \gN$.
The set $\tTree_{\gGG}(\gA)$ of (partial) \emph{derivation trees} of $\gA$ is defined inductively by the following rules:
\begin{enumerate}
	\item $\gA \in \tTree_{\gGG}(\gA)$ (i.e. derivations can be stopped prematurely).
	\item If $(\gA, \gR) \in \gP$ and $\tAss$ is a function that assigns to every $\he \in \hE_{\gR}$ a derivation tree $\tAss(\he) \in \tTree_{\gGG}(\hL_{\gR}(\he))$, then $(\gR, \tAss) \in \tTree_{\gGG}(\gA)$.
\end{enumerate}
A derivation tree is \emph{complete} if it is constructed without rule (1). Furthermore, it is a \emph{tree with hole} $\gA$ if (1) is applied exactly once and with nonterminal $\gA$.
\end{definition}
%
Every partial derivation tree describes how to apply production rules to derive an HG from a nonterminal through hyperedge replacement.
Formally, the HG $\sem{\tT}$ assembled from a derivation tree $\tT \in \tTree_{\gGG}(\gA)$ is given by
\begin{align*}
\sem{\tT} ~=~
\begin{cases}
  \handle{\gA}, & \tT = \gA \\
  \repl{\gR}{\lambda \he. \sem{\tAss(\he)}}, & \tT = (\gR, \tAss).
\end{cases}
\end{align*}
The first case stops further replacements of a hyperedge labeled $\gA$. Hence, complete derivation trees and trees with a hole yield HGs with zero and one hyperedge, respectively.
The second case replaces each hyperedge by the HG assembled from its assigned derivation tree.
Our formal model checking problem will be phrased in terms of the \emph{language} of an \HRG, i.e. all hypergraphs that are assembled from complete derivation trees of a starting symbol.
\begin{definition}
The \emph{language} of an \HRG $\gGG = (\gN, \gZ, \gP)$ and some $\gA \in \gN$ is 
\[
\lLL_\gA(\gGG) ~= \left\{ \sem{\tT} \middle\bracevert \tT \in \tTree_{\gGG}(\gA)~\text{complete} \right\}.
\]
Moreover, the \emph{language} of HGs \emph{generated} by $\gGG$ is $\lLL(\gGG) = \bigcup_{\gA \in \gZ} \lLL_\gA(\gGG)$.
\end{definition}
\begin{figure}[t]
\centering	
\begingroup\def\sx{0.75}
$$
\lLL_\gA(\gGG) = \left\{
\begin{matrix} \gR_1 \\
\scalebox{0.7}{\begin{tikzpicture}[thick,shorten >=6pt,shorten <=6pt]
    \draw[->] (0*\sx,0) to[out=-60,in=-120,distance=15] (1*\sx,0);
    \draw[->] (1*\sx,0) to[out=+120,in=+60,distance=15] (0*\sx,0);
    \draw (0*\sx,0) node[ver1,circle,inp] {$\mathbf 1$};
    \draw (1*\sx,0) node[ver1,circle,inp] {$\mathbf 2$};
\end{tikzpicture}} \end{matrix}
\;,\;\;
\begin{matrix} \repl{\gR_2}{\gR_1} \\
\scalebox{0.7}{\begin{tikzpicture}[thick,shorten >=6pt,shorten <=6pt]
    \draw[->] (0*\sx,0) to[out=-60,in=-120,distance=15] (1*\sx,0);
    \draw[->] (1*\sx,0) to[out=+120,in=+60,distance=15] (0*\sx,0);
    \draw[->] (1*\sx,0) to[out=-60,in=-120,distance=15] (2*\sx,0);
    \draw[->] (2*\sx,0) to[out=+120,in=+60,distance=15] (1*\sx,0);
    \draw (0*\sx,0) node[ver1,circle,inp] {$\mathbf 1$};
    \draw (1*\sx,0) node[ver1,circle,col=red] {};
    \draw (2*\sx,0) node[ver1,circle,inp] {$\mathbf 2$};
\end{tikzpicture}} \end{matrix}
\;,\;\; \dots\right\}
\qquad
\lLL_\gS(\gGG) = \left\{
\begin{matrix} \repl{\gR_3}{\gR_1} \\
\scalebox{0.7}{\begin{tikzpicture}[thick,shorten >=6pt,shorten <=6pt]
    \draw[->] (0*\sx,0) to[out=-60,in=-120,distance=15] (1*\sx,0);
    \draw[->] (1*\sx,0) to[out=+120,in=+60,distance=15] (0*\sx,0);
    \draw[->] (1*\sx,0) to[out=-60,in=-120,distance=15] (2*\sx,0);
    \draw[->] (2*\sx,0) to[out=+120,in=+60,distance=15] (1*\sx,0);
    \draw (0*\sx,0) node[ver1,circle,col=red] {};
    \draw (1*\sx,0) node[ver1,circle,col=red] {};
    \draw (2*\sx,0) node[ver1,circle,col=blue] {};
\end{tikzpicture}} \end{matrix}
\;,\;\;
\begin{matrix} \repl{\gR_3}{\repl{\gR_2}{\gR_1}} \\
\scalebox{0.7}{\begin{tikzpicture}[thick,shorten >=6pt,shorten <=6pt]
    \draw[->] (0*\sx,0) to[out=-60,in=-120,distance=15] (1*\sx,0);
    \draw[->] (1*\sx,0) to[out=+120,in=+60,distance=15] (0*\sx,0);
    \draw[->] (1*\sx,0) to[out=-60,in=-120,distance=15] (2*\sx,0);
    \draw[->] (2*\sx,0) to[out=+120,in=+60,distance=15] (1*\sx,0);
    \draw[->] (2*\sx,0) to[out=-60,in=-120,distance=15] (3*\sx,0);
    \draw[->] (3*\sx,0) to[out=+120,in=+60,distance=15] (2*\sx,0);
    \draw (0*\sx,0) node[ver1,circle,col=red] {};
    \draw (1*\sx,0) node[ver1,circle,col=red] {};
    \draw (2*\sx,0) node[ver1,circle,col=red] {};
    \draw (3*\sx,0) node[ver1,circle,col=blue] {};
\end{tikzpicture}} \end{matrix}
\;,\;\; \dots\right\}
$$
\endgroup
\vspace{-1.5em}
\caption{Languages of the grammar from \Cref{fig:grammar-dll}.}
\label{fig:grammar-lang}
\end{figure}
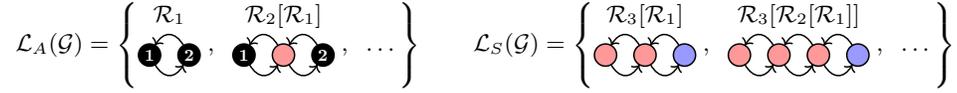
\Cref{fig:grammar-lang} illustrates the language of the HRG $\gGG$ from \Cref{fig:grammar-dll} for the nonterminals $\gA$ and $\gS$.
Since $\arity(\gA) = 2$, every HG in $\lLL_\gA(\gGG)$ has two abstract nodes.
Every HG in $\lLL_\gS(\gGG)$ is an LTS modeling a doubly-linked list, where the rightmost node is blue and every other node is red.

\subsection{Formal Problem Statement}\label{sec:problem}
We now formalize the main problem studied in this paper and discuss applications to model checking and satisfiability of infinite families of finite LTSs.

From now on, we will only work with HRGs $\gGG$ whose language $\lLL(\gGG)$ consists of LTSs, which is guaranteed if $\arity(\gS) = 0$ holds for every starting symbol $\gS$.
Furthermore, we denote by $\cmDel{c}{\hH}$ the HG that is identical to $\hH$ except that the node color $c \in \hCC$ has been deleted from every node.\footnote{i.e. $\cmDel{c}{\hH}$ is identical to $\hH$ except for the node coloring function $\lambda v.~ \hL_{\hH}(v) \setminus \{c\}$.}
Our goal is then to determine, for every LTS $\hH$ generated by an HRG, all nodes $v$ that can be used as initial states of $\hH$ such that $(\hH,v)$ satisfies a given $\omega$-regular property given by a B\"uchi automaton $\bMM$.
We characterize those nodes by constructing a \emph{recolored} HRG in which precisely those nodes are additionally colored with $\bMM$.

\begin{center}\noindent\fbox{\parbox{0.9\textwidth}{
\textbf{B\"uchi recoloring problem for HRGs.} 
Given an \HRG $\gGG$ generating LTSs with node colors $\hCC$ and a B\"uchi automaton $\bMM \in \bBB[2^\hCC]$, compute an \HRG $\gGG_{\bMM}$ using the set of node colors $\hCC \uplus \{\bMM\}$ such that 
\begin{enumerate}
	\item[a)] $\cmDel{\bMM}{\lLL(\gGG_{\bMM})} \,=\, \{ \cmDel{\bMM}{\hH} ~|~ \hH \in \lLL(\gGG_{\bMM}) \} \,=\, \lLL(\gGG)$ and 
	\item[b)] $\forall \hH \in \lLL(\gGG_\bMM)~ \forall \hv \in \hV_{\hH}.\quad \bMM \in \hC_{\hH}(\hv) \iff \hH,\hv \models \bMM$.
\end{enumerate}
}}\end{center}
Property (a) requires that the recolored HRG generates the same LTSs apart from adding the color $\bMM$ to some nodes
and property (b) requires that the right nodes are colored with $\bMM$.
If we use a fixed node color $init$ to indicate initial states of interest in the LTSs generated by an \HRG $\gGG$,
then a solution to the above problem allows deciding whether \emph{all}, \emph{any}, or \emph{infinitely many} LTSs generated by $\gGG$ satisfy an $\omega$-regular property given by a B\"uchi automaton $\bMM$.
For example, to reason about \emph{all} LTSs, it suffices to (1) compute the recolored \HRG $\gGG_{\bMM}$, (2) remove all unproductive rules\footnote{A rule is unproductive, it does not appear in any complete derivation tree of a starting symbol. Unproductive rules can be removed in linear time~(cf.~\cite{hopcroft2001introduction}).}, and (3) check whether for all remaining production rules that every node colored as initial is also colored with $\bMM$.
Decision procedures for the other problems work similarly, but need to find a single derivation tree or a loop among production rules in step (3) with the same property as above.
Finally, recoloring the grammar rather than deciding the $\omega$-regular property alone, allows extending our approach to \CTLS properties via recursive HRG recoloring; see \Cref{sec:impl} for details.

\begin{figure}[t]
\centering	
\begingroup\def\sx{0.75}
$$
\gCC_\gA(\gGG) = \left\{
\begin{matrix} \repl{\gR_3}{\handle{\gA}} \\
\scalebox{0.6}{\begin{tikzpicture}[thick,shorten >=6pt,shorten <=6pt]
    \draw     (0*\sx,0) -- (1*\sx,0) node[pos=0.5, above, yshift=-1.5pt] {$\mathbf 1$};
    \draw     (1*\sx,0) -- (2*\sx,0) node[pos=0.5, above, yshift=-1.5pt] {$\mathbf 2$};
    \draw[->] (2*\sx,0) to[out=-60,in=-120,distance=15] (3*\sx,0);
    \draw[->] (3*\sx,0) to[out=+120,in=+60,distance=15] (2*\sx,0);
    \draw (0*\sx,0) node[ver1,circle,col=red] {};
    \draw (1*\sx,0) node[ver1,mod] {$\gA$};
    \draw (2*\sx,0) node[ver1,circle,col=red] {};
    \draw (3*\sx,0) node[ver1,circle,col=blue] {};
\end{tikzpicture}} \end{matrix}
\;,\;\;
\begin{matrix} \repl{\gR_3}{\repl{\gR_2}{\handle{\gA}}} \\
\scalebox{0.6}{\begin{tikzpicture}[thick,shorten >=6pt,shorten <=6pt]
    \draw[->] (0*\sx,0) to[out=-60,in=-120,distance=15] (1*\sx,0);
    \draw[->] (1*\sx,0) to[out=+120,in=+60,distance=15] (0*\sx,0);
    \draw     (1*\sx,0) -- (2*\sx,0) node[pos=0.5, above, yshift=-1.5pt] {$\mathbf 1$};
    \draw     (2*\sx,0) -- (3*\sx,0) node[pos=0.5, above, yshift=-1.5pt] {$\mathbf 2$};
    \draw[->] (3*\sx,0) to[out=-60,in=-120,distance=15] (4*\sx,0);
    \draw[->] (4*\sx,0) to[out=+120,in=+60,distance=15] (3*\sx,0);
    \draw (0*\sx,0) node[ver1,circle,col=red] {};
    \draw (1*\sx,0) node[ver1,circle,col=red] {};
    \draw (2*\sx,0) node[ver1,mod] {$\gA$};
    \draw (3*\sx,0) node[ver1,circle,col=red] {};
    \draw (4*\sx,0) node[ver1,circle,col=blue] {};
\end{tikzpicture}} \end{matrix}
\;,\;\; 
\begin{matrix} \repl{\gR_3}{\repl{\gR_2}{\repl{\gR_2}{\handle{\gA}}}} \\
\scalebox{0.6}{\begin{tikzpicture}[thick,shorten >=6pt,shorten <=6pt]
    \draw[->] (0*\sx,0) to[out=-60,in=-120,distance=15] (1*\sx,0);
    \draw[->] (1*\sx,0) to[out=+120,in=+60,distance=15] (0*\sx,0);
    \draw[->] (1*\sx,0) to[out=-60,in=-120,distance=15] (2*\sx,0);
    \draw[->] (2*\sx,0) to[out=+120,in=+60,distance=15] (1*\sx,0);
    \draw     (2*\sx,0) -- (3*\sx,0) node[pos=0.5, above, yshift=-1.5pt] {$\mathbf 1$};
    \draw     (3*\sx,0) -- (4*\sx,0) node[pos=0.5, above, yshift=-1.5pt] {$\mathbf 2$};
    \draw[->] (4*\sx,0) to[out=-60,in=-120,distance=15] (5*\sx,0);
    \draw[->] (5*\sx,0) to[out=+120,in=+60,distance=15] (4*\sx,0);
    \draw (0*\sx,0) node[ver1,circle,col=red] {};
    \draw (1*\sx,0) node[ver1,circle,col=red] {};
    \draw (2*\sx,0) node[ver1,circle,col=red] {};
    \draw (3*\sx,0) node[ver1,mod] {$\gA$};
    \draw (4*\sx,0) node[ver1,circle,col=red] {};
    \draw (5*\sx,0) node[ver1,circle,col=blue] {};
\end{tikzpicture}} \end{matrix}
\;,\;\; \dots\right\}
, \quad
\gCC_\gS(\gGG) = \left\{
\begin{matrix}  \handle{\gS} \\
\scalebox{0.6}{\begin{tikzpicture}[thick,shorten >=6pt,shorten <=6pt]
    \draw (0*\sx,0) node[ver1,mod] {$\gS$};
\end{tikzpicture}} \end{matrix}
\right\}
$$
\endgroup
\vspace{-1.5em}
\caption{Contexts of the grammar from \Cref{fig:grammar-dll}.}
\label{fig:grammar-ctx}
\end{figure}

\section{Compositional Reasoning about HRG Languages} \label{sec:refinement}
This section presents a generic approach for compositional reasoning about \HRGs based on congruences, which we apply in \Cref{sec:meq} to solve the B\"uchi recoloring problem.
The main idea is that every complete derivation tree $\tT \in \tTree_{\gGG}(\gA)$ can be decomposed into three components:
(1) a tree $\tT' \in \tTree_{\gGG}(\gA)$ with some hole $\gB$, (2) a production rule $(\gB,\gR) \in \gP_{\gGG}$, and (3) a complete derivation tree for every hyperedge in $\gR$ (cf.~\cite{engelfriet1997context}).
Our goal is to reason compositionally about the HGs $\sem{\tT}$ of such trees $\tT$ by considering only component (2), i.e. individual production rules $(\gB,\gR)$, and limited information on the components (1) and (3).
The HGs obtained from (1) are also called contexts.
\begin{definition}
The set of \emph{contexts} of an \HRG $\gGG = (\gN, \gZ, \gP)$ and $\gB \in \gN$ is
\[
\gCC_\gB(\gGG) ~= \left\{\sem{\tT} \middle\bracevert  \exists \gA \in \gZ.\; \tT \in \tTree_\gGG(\gA)~\text{with hole}~\gB \right\}.
\]
\end{definition}
\Cref{fig:grammar-ctx} illustrates the sets of contexts of the HRG $\gGG$ from \Cref{fig:grammar-dll}.
Note that there is only one context of $\gS$, because $\gGG$ contains no hyperedge labelled $\gS$.
The above decomposition of complete derivation trees can also be phrased as a decomposition of the HGs assembled from components (1) - (3).
\begin{replemma}{lem:hg-decomposition}
Let $\gGG$ be a \HRG, $\hH \in \lLL(\gGG)$, and $\hv \in \hV_{\hH}$.
Then, there exists a rule $(\gB,\gR) \in \gP_{\gGG}$, a context $\hJ \in \gCC_{\gB}(\gGG)$, and an assignment $\hAss\colon \hE_{\hH} \to \hHH$ with $\hAss(\he)\in \lLL_{\hL_{\hH}(\he)}(\gGG)$ for all $\he \in \hE_{\hH}$ such that
$\hH = \repl{\hJ}{\repl{\gR}{\fml{\hK}}}$ and $\hv \in \hV_{\gR}$.
\end{replemma}
To formalize what we mean by limited information about (1) and (3), we define \emph{HGs with a view}, i.e. HGs with a distinguished set of \emph{exposed} nodes. We assume that abstract nodes and attached nodes are always exposed, as they are affected by hyperedge replacement. We allow exposing an additional subset of the remaining unattached nodes $\cmU_{\hH} = \{ v \in \hV_{\hH} ~|~ \forall e \in \hE_{\hH}. v \notin \hd_{\hH}(\he) \}$. 
\begin{definition}
A \emph{hypergraph with a view} (HGV) is a pair $(\hH,U)$, where $\hH \in \hHH$ and $U \subseteq \cmU_\hH$. We denote by $X_\hVV$ the set of \emph{exposed nodes} $U \cup (\hW_\hH \setminus \cmU_\hH)$. We call $(\hH,\emptyset)$ the \emph{trivial HGV} of $\hH$ and $(\hH,\cmU_\hH)$ the \emph{full HGV} of $\hH$.
\end{definition}
We lift hyperedge replacement to HGVs by joining views, i.e.
\begin{align*}
	\repl{(\hH,U)}{\lambda \he.\; (\fml\hJ(\he), U_\he)} = (\repl{\hH}{\fml\hJ}, \textstyle{\biguplus_\he} U_\he \uplus U).
\end{align*}
Analogously, we lift paths and traces to HGVs.
Furthermore, for a B\"uchi automaton $\bMM$, $\hVV = (\hH,U)$, and $x \in X_{\hVV}$, we define $\hVV, x \models \bMM$ iff $\hH, x \models \bMM$.
We characterize the information about the context $\hJ$ and the HGs $\hAss(\he)$ needed for compositional reasoning as equivalence classes of relations between HGVs. 
To this end, we must relate exposed nodes and hyperedges of HGVs.
\begin{definition}
Let $\hVV=(\hH,U) \in \hHHV$ and $\hVV' =(\hH',U') \in \hHHV$. A \emph{coupling} between $\hVV$ and $\hVV'$ is a pair $(\eta,\mu)$ where $\eta: U \bij U'$ and $\mu: \hE_\hH \bij \hE_{\hH'}$ such that there exists a bijection $\nu$ for which
\begin{align*}
\forall \hi \in \hI_\hH. \nu(\hi) = \hi,
\quad 
\forall \hv \in X. \nu(\hv) = \eta(\hv),
\quad 
\forall \he \in \hE_\hH \forall i. \nu(\hd_{\hH}(\he)_i) = \hd_{\hH'}(\mu(\he))_i.
\end{align*}
\end{definition}
The mappings $\eta$ and $\mu$ determine how unattached nodes and hyperedges are related.
The relation between abstract nodes in the two HGs is  implicitly given by their numbering. The same holds for nodes attached to hyperedges.
We extend $\eta$ to a mapping $\nu$ that makes those correspondences explicit.
\begin{definition} \label{def:congr}
An HGV-relation $\equiv$ is a set of tuples $(\hUU,\hVV,\eta,\mu)$, where $\hUU,\hVV$ are HGVs and $(\eta,\mu)$ is a coupling.
We write $\eta,\mu \vdash \hUU \equiv \hVV$ instead of $(\hUU,\hVV,\eta,\mu) \in\, \equiv$.
\end{definition}
For example, isomorphisms are HGV-relations. Specifically, two views $(\hH,U)$ and $(\hH',U')$ are isomorphic under $\eta$ and $\mu$, written $\eta,\mu \vdash (\hH,U) \cong (\hH',U')$, iff $\hA_{\hH'} = \nu(\hA_{\hH})$, $\hC_\hH = \hC_{\hH'} \circ \eta$ and $\hL_\hH = \hL_{\hH'} \circ \mu$. The hypergraphs $\hH$ and $\hH'$ are isomorphic if their full HGVs are isomorphic.

The precise choice of $\eta$ and $\mu$ in HGV-relations $\equiv$ is often irrelevant.
Hence, we define
$\hUU \equiv \hVV
~\text{iff}~ \exists\eta,\mu.~ \eta,\mu \vdash \hUU \equiv \hVV
$
that uses the same symbol $\equiv$ by slight abuse of notation.
We will only work with HGV-relations $\equiv$ that are equivalence relations and respect hyperedge replacement, i.e.
replacing hyperedges by equivalent HGVs in equivalent HGVs must yield equivalent HGVs again. 
\begin{definition} \label{def:congr}
An \emph{HGV-congruence} is an HGV-relation $\equiv$ such that, for all $\hUU,\hUU' \in \hHHV$ and hyperedge assignments $\fml\hVV$ on $\hUU$ and $\fml\hVV'$ on $\hUU'$,
\begin{align*}
&	(\eta,\mu) \vdash \hUU \equiv \mathcal \hUU' \quad\text{and}\quad\forall \he \in \hE_\hH.\; (\eta_\he, \mu_\he) \vdash \fml\hVV(\mu(\he)) \equiv \fml\hVV'(\mu(\he))
\\
& \textit{implies}\quad
	({\textstyle\biguplus_\he} \eta_\he \uplus \eta,\; {\textstyle\biguplus_\he} \mu_\he) \vdash \repl\hUU{\fml\hVV} \equiv \repl{\hUU'}{\fml\hVV'}.
\end{align*}
\end{definition}
For instance, every partial isomorphism 
is an HGV-congruence. If $\equiv$ is an HGV-congruence, we denote by $\eqclass{\hVV}_\equiv$ the equivalence class of $\hVV$.
For every \HRG $\gGG$ and every HGV-congruence $\equiv$, one can then construct a \emph{refined \HRG} $\gGG_\equiv$.
For every rule $(\gA, \gR)$ that appears in $\hH = \repl{\hJ}{\repl{\gR}{\fml\hK}} \in \lLL(\gGG)$, $\gGG_\equiv$ annotates $\gA$ with the equivalence class $\eqclass{(\hJ, \emptyset)}$, and each $\hL_\gR(\he)$ with $\eqclass{(\fml\hK(\he), \emptyset)}$.
	These equivalence classes allow computing $\eqclass{(\repl{\hJ}{\repl{\gR}{\fml\hK}}, \hV_\gR)}$ without knowing $\hJ$ or $\hK$. We can thus reason about $\gR$ in the context in which it is applied.

To formalize refined \HRGs, we rely on two conventions.
First, we denote by $\grule{\gR}{\fml{\gB}}{\gA}$ the production rule $(\gA,\gR')$, where $\gR'$ is identical to $\gR$ except that the hyperedge labeling function is $\fml{\gB}$ instead of $\hL_{\gR}$.
Second, given two functions $f,g$ with the same domain, we denote by $f \times g$ the function that assigns the value of both functions to each domain element, i.e. $f \times g = \lambda e. (f(e), g(e))$. %
\begin{definition}\label{def:refine}
Let $\gGG = (\gN,\gP,\gZ)$ be an \HRG and $\equiv$ be a HGV-congruence.
Then the \emph{refined \HRG} $\gGG_{\equiv} = (\gN_{\equiv},\gP_{\equiv},\gZ_{\equiv})$ is defined as follows:
	\begin{align*}
		\gN_{\equiv} &= \{(\gA,\gs,\gx) \mid \gA \in \gN,\; \gs \in \eqclass{(\lLL_{\gA}(\gGG),\emptyset)},\; \gx \in \eqclass{(\gCC_{\gA}(\gGG),\emptyset)} \} \\
		\gP_{\equiv} &= \left\{\grule{\gR}{\fml\gB \zip \fml\gr \zip \fml\gy}{(\gA,\gs,\gx)} \middle\bracevert \substack{\grule{\gR}{\fml{\gB}}{\gA} \in \gP,\; \gs = \repl{(\gR,\emptyset)}{\fml{\gr}}, \\ \fml\gy = \lambda \he.\; \repl{\gx}{\repl{(\gR,\emptyset)}{\fml\gr\langle \he / \handle{\fml\gB(\he)} \rangle}}} \right\} \\
		\gZ_{\equiv} &= \{(\gA,\gs,\eqclass{(\handle{\gA},\emptyset)}) \in \gN_{\equiv} \mid \gA \in \gZ\},
	\end{align*}
	where $\fml\gr\langle\he / \handle{\gB(\he)}\rangle$ is the hyperedge assignment that is identical to $\fml\gr$ except that it assigns the handle $\handle{\gB(\he)}$ to $\he$.
\end{definition}
Note that the information about contexts and languages of hyperedges added in the above construction is \emph{local} to each rule, i.e. it depends solely on the equivalence classes assigned to the nonterminal of the considered rule and the nonterminals assigned to hyperedges in that rule.\footnote{\new{We will illustrate refined HRGs at the end of the next section, which introduces a concrete HGV-congruence for reasoning about $\omega$-regular properties.}}
The following theorem justifies that enriching \HRGs with such information is correct for HGV-congruences.
\begin{reptheorem}{thm:ref-ok}
Let $\gGG$ be an \HRG and $\equiv$ be a HGV-congruence. Then:
\begin{enumerate}
	\item For all $\hH \in \lLL_{(\gA,\gs,\gx)}(\gGG_{\equiv})$, we have $\gs = \eqclass{(\hH,\emptyset)}$.
	\item For all $\hJ \in \gCC_{(\gA,\gs,\gx)}(\gGG_{\equiv})$, we have $\gx = \eqclass{(\hJ,\emptyset)}$.
\end{enumerate}
\end{reptheorem}
In other words, the refined \HRG can only generate HGs and contexts from a nonterminal $(\gA,\gs,\gx)$ that are in the equivalence classes $\gs$ and $\gx$, respectively. 
\section{$\bMM$-equivalence: A Congruence for B\"uchi Recoloring} \label{sec:meq}

To solve the B\"uchi recoloring problem (see \Cref{sec:problem}), we now define an HGV-congruence for $\omega$-regular properties.
Throughout this section, we fix a B\"uchi automaton $\bMM = (\bQ, \bS, \bD, \bI, \bF)$ with $\bS = 2^{\hCC}$.
Our goal is to define a congruence relation on HGVs $\hVV$ that suffices to decide whether $\bMM$ accepts on the traces that start and end on exposed nodes $X_\hVV$. 
We first define such an equivalence on strings, which resembles the indistinguishability relation used in the proof of the well-known Myhill-Nerode theorem for (cf.~\cite{hopcroft2001introduction,thomas1990automata}). 

\begin{definition}[$\bMM$-equivalence for Strings] \label{def:string-meq}
Two strings $\seq\bs,\seq\bs' \in \bS^*$ are \emph{$\bMM$-equivalent} if $\forall \bp,\bq,\boolb.\; (\bp,\seq\bs,\bq,\boolb) \in \bD^* \iff (\bp,\seq\bs',\bq,\boolb) \in \bD^*$.
Two infinite strings $\seq\bs,\seq\bs' \in \bS^\omega$ are \emph{$\bMM$-equivalent} if
$
	\forall \bp.\; (\bp,\seq\bs,\boolT) \in \bD^\omega \iff (\bp,\seq\bs',\boolT) \in \bD^\omega.
$
We write $\seq\bs \equiv_\bMM \seq\bs'$ to denote that $\seq\bs$ and $\seq\bs'$ are $\bMM$-equivalent. 
\end{definition}
$\bMM$-equivalent strings cannot be distinguished by $\bMM$. 
We lift this notion to sets of strings. Specifically, two sets of finite (resp. infinite) strings $S,S' \subseteq \Sigma^*$ (resp. $S,S' \subseteq \Sigma^\omega$) are equivalent, iff for each string $\seq\bs \in S$, there exists an $\bMM$-equivalent string $\seq\bs' \in S'$, and vice versa.
To lift $\bMM$-equivalence to HGVs $\hVV$, we consider the traces between all exposed nodes $x,y \in X_\hVV$.
For every equivalent $\hVV'$, we then require that the sets of traces between coupled exposed nodes are equivalent.
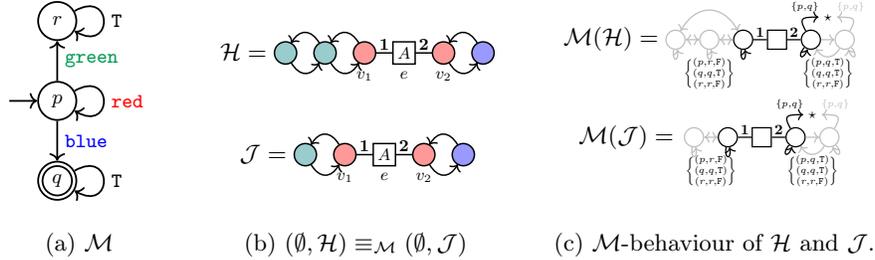
\begin{figure}[t]
\begin{subfigure}[b]{0.2\linewidth}
\centering
\scalebox{0.85}{\begin{tikzpicture}[thick]
    \node (P) at (0,0) [ver,circle] {$\bp$};
    \node (Q) at (0,-1.25) [ver,circle] {$\bq$};
    \draw (Q) node[circle,draw, scale=1.5] {};
    \node (R) at (0,1.25) [ver,circle] {$\br$};

    \draw[->] (-0.75,0) -- (P);
    \draw[->] (P) -- node[pos=0.5, right] {$\color{blue} \cgreen$} (R);
    \draw[->] (P) -- node[pos=0.5, right] {$\color{blue} \cblue$} (Q);
    \draw[->] (P) to[out=30,in=330,distance=20] node[pos=0.5, right] {$\cred$} (P); 
    \draw[->] (Q) to[out=30,in=330,distance=20] node[pos=0.5, right] {$\boolT$} (Q);  
    \draw[->] (R) to[out=30,in=330,distance=20] node[pos=0.5, right] {$\boolT$} (R);  
\end{tikzpicture}}
\caption{$\bMM$}
\label{fig:meq-ex-a}
\end{subfigure}	
\hfill
\begin{subfigure}[b]{0.38\linewidth}
\centering
\begingroup\def\sx{0.75}
$$
\renewcommand{\arraystretch}{3.5}
\begin{array}{c@{}}
\hH = \scalebox{0.7}{\begin{tikzpicture}[thick,shorten >=6pt,shorten <=6pt, baseline=-3.5pt]
    \draw[->] (0*\sx,0) to[out=-60,in=-120,distance=15] (1*\sx,0);
    \draw[->] (1*\sx,0) to[out=+120,in=+60,distance=15] (0*\sx,0);
    \draw[->] (1*\sx,0) to[out=-60,in=-120,distance=15] (2*\sx,0);
    \draw[->] (2*\sx,0) to[out=+120,in=+60,distance=15] (1*\sx,0);
    \draw     (2*\sx,0) -- (3*\sx,0) node[pos=0.5, above, yshift=-1.5pt] {$\mathbf 1$};
    \draw     (3*\sx,0) -- (4*\sx,0) node[pos=0.5, above, yshift=-1.5pt] {$\mathbf 2$};
    \draw[->] (4*\sx,0) to[out=-60,in=-120,distance=15] (5*\sx,0);
    \draw[->] (5*\sx,0) to[out=+120,in=+60,distance=15] (4*\sx,0);
    \draw (0*\sx,0) node[ver1,circle,col=teal] {};
    \draw (1*\sx,0) node[ver1,circle,col=teal] {};
    \draw (2*\sx,0) node[ver1,circle,col=red, label=270:$\hv_1$] {};
    \draw (3*\sx,0) node[ver1,mod, label=270:$\he$] {$\gA$};
    \draw (4*\sx,0) node[ver1,circle,col=red, label=270:$\hv_2$] {};
    \draw (5*\sx,0) node[ver1,circle,col=blue] {};
\end{tikzpicture}} 
\\
\hJ = \scalebox{0.7}{\begin{tikzpicture}[thick,shorten >=6pt,shorten <=6pt, baseline=-3.5pt]
    \draw[->] (0*\sx,0) to[out=-60,in=-120,distance=15] (1*\sx,0);
    \draw[->] (1*\sx,0) to[out=+120,in=+60,distance=15] (0*\sx,0);
    \draw     (1*\sx,0) -- (2*\sx,0) node[pos=0.5, above, yshift=-1.5pt] {$\mathbf 1$};
    \draw     (2*\sx,0) -- (3*\sx,0) node[pos=0.5, above, yshift=-1.5pt] {$\mathbf 2$};
    \draw[->] (3*\sx,0) to[out=-60,in=-120,distance=15] (4*\sx,0);
    \draw[->] (4*\sx,0) to[out=+120,in=+60,distance=15] (3*\sx,0);
    \draw (0*\sx,0) node[ver1,circle,col=teal] {};
    \draw (1*\sx,0) node[ver1,circle,col=red, label=270:$\hv_1$] {};
    \draw (2*\sx,0) node[ver1,mod, label=270:$\he$] {$\gA$};
    \draw (3*\sx,0) node[ver1,circle,col=red, label=270:$\hv_2$] {};
    \draw (4*\sx,0) node[ver1,circle,col=blue] {};
\end{tikzpicture}}
\end{array}
$$
\endgroup
\caption{$(\emptyset, \hH) \equiv_\bMM (\emptyset, \hJ)$}
\label{fig:meq-ex-g}
\end{subfigure}
\hfill
\begin{subfigure}[b]{0.38\linewidth}
\begingroup\def\sx{0.75}
$$
\renewcommand{\arraystretch}{2}
\begin{array}{c@{}}
\bMM(\hH) = \scalebox{0.6}{\begin{tikzpicture}[thick, baseline=-3.5pt]
    \node (V1) at (0*\sx,0) [lightgray,ver1,circle] {};
    \node (V2) at (1*\sx,0) [lightgray,ver1,circle] {};
    \node (V3) at (2*\sx,0) [ver1,circle] {};
    \node (E1) at (3*\sx,0) [ver1,mod] {};
    \node (V4) at (4*\sx,0) [ver1,circle] {};
    \node (V5) at (5*\sx,0) [lightgray,ver1,circle] {};
    \node (STAR) at (4.5*\sx,0.5) [color=white,text=black] {$\star$};
    
    \draw     (E1) -- (V3) node[pos=0.5, above, yshift=-1.5pt] {$\mathbf 1$};
    \draw     (E1) -- (V4) node[pos=0.5, above, yshift=-1.5pt] {$\mathbf 2$};
   	\draw[->] (V4) to[out=230,in=250,distance=8] (V4); 
   	\draw[lightgray,->] (V5) to[out=310,in=290,distance=8] (V5); 
    \draw[lightgray,->] (V4) -- (V5); 
    \draw[lightgray,->] (V5) to[out=240,in=300,distance=8] (V4);
    \draw[lightgray,<->] (V1) -- (V2); 
    \draw[lightgray,<->] (V2) -- (V3); 
    \draw[lightgray,<->] (V1) to[out=60,in=120,distance=20] (V3);
    \draw[lightgray,->] (V1) to[out=260,in=280,distance=8] (V1); 
   	\draw[lightgray,->] (V2) to[out=260,in=280,distance=8] (V2); 
   	\draw[->] (V3) to[out=260,in=280,distance=8] (V3); 

    \draw[->] (V4) to[out=120,in=180,distance=10] node[pos=0.6, above, yshift=1pt] {$\scriptstyle \{\bp,\bq\}$} (STAR); 
    \draw[lightgray,->] (V5) to[out=60,in=0,distance=10] node[pos=0.6, above, yshift=1pt] {$\scriptstyle \{\bp,\bq\}$} (STAR); 
    
    \node at (V2) [black, below, yshift=-8pt] {$\scriptstyle \left\{\substack{(\bp,\br,\boolF)\\(\bq,\bq,\boolT)\\(\br,\br,\boolF)}\right\}$};
    \node at ($(V4)!0.5!(V5)$) [black, below, yshift=-8pt] {$\scriptstyle \left\{\substack{(\bp,\bq,\boolT)\\(\bq,\bq,\boolT)\\(\br,\br,\boolF)}\right\}$};
\end{tikzpicture}} 
\\
\bMM(\hJ) = \scalebox{0.6}{\begin{tikzpicture}[thick, baseline=-3.5pt]
    \node (V2) at (1*\sx,0) [lightgray,ver1,circle] {};
    \node (V3) at (2*\sx,0) [ver1,circle] {};
    \node (E1) at (3*\sx,0) [ver1,mod] {};
    \node (V4) at (4*\sx,0) [ver1,circle] {};
    \node (V5) at (5*\sx,0) [lightgray,ver1,circle] {};
    \node (STAR) at (4.5*\sx,0.5) [color=white,text=black] {$\star$};
    
    \draw     (E1) -- (V3) node[pos=0.5, above, yshift=-1.5pt] {$\mathbf 1$};
    \draw     (E1) -- (V4) node[pos=0.5, above, yshift=-1.5pt] {$\mathbf 2$};
   	\draw[->] (V4) to[out=230,in=250,distance=8] (V4); 
   	\draw[lightgray,->] (V5) to[out=310,in=290,distance=8] (V5); 
    \draw[lightgray,->] (V4) -- (V5); 
    \draw[lightgray,->] (V5) to[out=240,in=300,distance=8] (V4);
    \draw[lightgray,<->] (V2) -- (V3); 
   	\draw[lightgray,->] (V2) to[out=260,in=280,distance=8] (V2); 
   	\draw[->] (V3) to[out=260,in=280,distance=8] (V3); 

    \draw[->] (V4) to[out=120,in=180,distance=10] node[pos=0.6, above, yshift=1pt] {$\scriptstyle \{\bp,\bq\}$} (STAR); 
    \draw[lightgray,->] (V5) to[out=60,in=0,distance=10] node[pos=0.6, above, yshift=1pt] {$\scriptstyle \{\bp,\bq\}$} (STAR); 
    
    \node at ($(V2)!0.5!(V3)$) [black, below, yshift=-8pt] {$\scriptstyle \left\{\substack{(\bp,\br,\boolF)\\(\bq,\bq,\boolT)\\(\br,\br,\boolF)}\right\}$};
    \node at ($(V4)!0.5!(V5)$) [black, below, yshift=-8pt] {$\scriptstyle \left\{\substack{(\bp,\bq,\boolT)\\(\bq,\bq,\boolT)\\(\br,\br,\boolF)}\right\}$};
\end{tikzpicture}}
\end{array}
$$
\endgroup	
\caption{$\bMM$-behaviour of $\hH$ and $\hJ$.}
\label{fig:meq-ex-m}
\end{subfigure}
\vspace{-0.5em}
\caption{B\"uchi automaton $\bMM$ for $(\cred\ U\ \cblue)$ and its behaviour on HGs.}
\label{fig:meq-ex}
\end{figure}
\begin{definition}[$\bMM$-equivalence for views]	 \label{def:graph-meq}
Let $\hVV,\hVV' \in \hHHV$ with coupling $(\eta,\mu)$ and mapping $\nu$.
Then $\hVV$ is \emph{$\bMM$-equivalent} to $\hVV'$ if all $\hu,\hv \in X_{\hVV}$ satisfy
\begin{align*}
	\Traces_\hVV^*(x,y) \equiv_\bMM \Traces_{\hVV'}^*(\nu(x),\nu(y))	~\text{and}~~ 
	\Traces_\hVV^\omega(x) \equiv_\bMM \Traces_{\hVV'}^\omega(\nu(x)).
\end{align*}
\end{definition}
\Cref{fig:meq-ex-a} shows a B\"uchi automaton and \Cref{fig:meq-ex-g} shows two example HGs. 
The HGVs $(\hH,\emptyset)$ and $(\hJ,\emptyset)$ are $\bMM$-equivalent. To see this, note that the two graphs have the same traces that start or end in $\hv_2$.
While the respective sets $\Traces^*(\hv_1,\hv_1)$ and $\Traces^\omega(\hv_1)$ differ for $\hH$ and $\hJ$, they are indistinguishable by $\bMM$ as all such paths transition $\bMM$ to the sink state $\br$.

$\bMM$-equivalence enjoys two important properties. 
First, it is a congruence.
\begin{reptheorem}{lem:meq-cong}
$\bMM$-equivalence is an HGV-congruence.
\end{reptheorem}
Second, for $\bMM$-equivalent graphs, the B\"uchi automaton $\bMM$ cannot distinguish between the traces of exposed nodes that are coupled (via $\nu$):
\begin{replemma}{lem:meq-sem}
Let $\hVV$ and $\hVV'$ be $\bMM$-equivalent via the mapping $\nu$.
Then, for all exposed nodes $x \in X_{\hVV}$, we have~~
$(\hVV,x) \models \bMM \iff (\hVV',\nu(x)) \models \bMM$.
\end{replemma}
The above lemma justifies lifting 
 $\models$ from HVGs to their $\bMM$-equivalence classes. 

We are now in a position to define a solution for the B\"uchi recoloring problem for a given HRG $\gGG$.
Morally, we consider every HG $\hH \in \lLL(\gGG)$ and any node $\hv \in \hV_{\hH}$.
To check whether $v$ needs to be colored with $\bMM$, we have to decide whether $\hH,\hv \models \bMM$ holds.
To this end, we decompose $\hH$ into $\repl{\hJ}{\repl{\gR}{\fml\hK}}$ such that $\hv \in \hV_{\gR}$. We can then extract the $\bMM$-equivalence class of $\hJ$ and each $\fml\hK(\he)$ from the refined grammar $\gGG_{\equiv\bMM}$ in order to compute the $\bMM$-equivalence class $\eqclass{(\repl{\hJ}{\repl{\gR}{\fml\hK}}, \hV_\gR)}$. \Cref{lem:meq-sem}, then guarantees that we can check whether the automaton $\bMM$ is satisfied on each exposed node $\hV_\gR$, including $\hv$. 

\begin{definition}[Grammar recoloring]\label{def:gram-recol}
Let $\gGG$ be an \HRG. Then, we define its $\bMM$-recoloring $\gGG_{\bMM} = (\gN_{\gGG_{\equiv\bMM}},\gP,\gZ_{\gGG_{\equiv\bMM}})$ where
$$
\gP = \left\{\grule{\gR'}{\fml\gB \zip \fml\gr \zip \fml\gy}{(\gA,\gs,\gx)} \middle\bracevert \substack{\grule{\gR}{\fml\gB \zip \fml\gr \zip \fml\gy}{(\gA,\gs,\gx)} \in \gP_{\gGG_{\equiv\bMM}} \land \gR = \cmDel{\bMM}{\gR'} \land \\\forall \hv \in \hV_{\gR'}.\; \bMM \in \hC_{\gR'}(\hv) \iff (\repl{\gx}{\repl{(\gR',\cmU_{\gR'})}{\fml{\gr}}}, \hv) \models \bMM} \right\}
$$
\end{definition}
\begin{reptheorem}[Soundness and Completeness]{thm:rec-gram-ok}
	Let $\gGG$ be a \HRG. Then,
	$$
		\forall \hH \in \lLL(\gGG_{\bMM}) \: \forall \hv \in \hV_{\hH}.\; \bMM \in \hC_{\hH}(\hv) \iff (\hH,\hv) \models \bMM.
	$$
\end{reptheorem}
\begin{proof}
Let $\hH \in \lLL(\gGG_\bMM)$, and let $\hv \in \hV_{\hH}$. By \Cref{lem:hg-decomposition}, there exists a rule $\grule\gR{\fml\gB \zip \fml\gr \zip \fml\gy}{(\gA,\gs,\gx)} \in \gP_{\gGG_\bMM}$, a context $\hJ \in \gCC_{(\gA,\gs,\gx)}(\gGG_\bMM)$ and an assignments $\fml{\hK} : \hE_\gR \rightarrow \hHH$ with each $\fml\hK(\he) \in \lLL_{(\fml\gB(\he),\fml\gr(\he),\fml\gy(\he))}(\gGG_\bMM)$ such that $\hH = \repl{\hJ}{\repl{\gR}{\fml{\hK}}}$ and $\hv \in \hV_{\gR}$. Thus,
\begin{align*}
	\bMM \in \hC_{\hH}(\hv) 
	  &\iff \bMM \in \hC_{\repl{\hJ}{\repl{\gR}{\fml{\hK}}}}(\hv) \tag{\Cref{lem:hg-decomposition}}
	\\&\iff \bMM \in \hC_{\gR}(\hv)  \tag{$v \in \hV_\gR$}
	\\&\iff (\repl{\gx}{\repl{(\gR, \cmU_\gR)}{\fml{\gs}}}, \hv) \models \bMM. \tag{\Cref{def:gram-recol}}
	\intertext{By \Cref{thm:ref-ok}, we have that $\gx$ is the equivalence class of $(\hJ, \emptyset)$ and $\fml\gs(\he)$ is the equivalence class of $(\fml\hK(\he), \emptyset)$. Then, by \Cref{lem:meq-cong}, $\repl{\gx}{\repl{\eqclass{(\gR, \cmU_\gR)}}{\fml{\gs}}}$ is the equivalence class of $(\repl{\hJ}{\repl{\gR}{\fml{\hK}}}, \hV_\gR)$. Hence, }
	&\iff (\repl{\hJ}{\repl{\gR}{\fml{\hK}}}, \hv) \models \bMM \tag{\Cref{lem:meq-sem}}
	\\&\iff (\hH, \hv) \models \bMM. \tag{\Cref{lem:hg-decomposition}, backwards}
\end{align*}
\end{proof}

\paragraph{Representing equivalence classes.}
In order to effectively check $\bMM$-equivalence, we introduce a canonical representation of HGs, called $\bMM$-behaviours.
\begin{definition}\label{def:mbhv}
The \emph{$\bMM$-behaviour} of a hypergraph $\hH$ is the HG \[\bMM(H) = (\hV_\hH \uplus \{\star\},\hn_\hH,{\rightsquigarrow^*} \uplus {\rightsquigarrow^\omega},\hE_\hH,\hd_\hH,\bot,\bot), \] where ${\rightsquigarrow^*} \subseteq \hW \times 2^{\bQ \times \bQ \times \Bool} \times \hW$ and ${\rightsquigarrow^\omega} \subseteq \hW \times 2^{\bQ} \times \{\star\}$ are given by
\begin{align*}
	\hu \rightsquigarrow_f^* \hv &\iff \exists \seq\hc \in \Traces_\hH^*(\hu,\hv).\; f = \{ (\bp,\bq,\boolb) \mid \bD^*(\bp, \seq\hc, \bq, \boolb) \} \\
	\hu \rightsquigarrow_f^\omega \star &\iff \exists \seq\hc \in \Traces_\hH^\omega(\hu).\; f = \{ \bp \mid \bD^\omega(\bp,\seq\hc,\boolT)\}.
\end{align*}
\end{definition}
\Cref{fig:meq-ex-m} depicts the $\bMM$-behaviours of each HG in \Cref{fig:meq-ex-g}.
Intuitively, an $\bMM$-behavior of $\hH$ captures, for every pair of nodes $\hu$, $\hv$ whether $\bMM$ can transition from state $\bp$ to $\bq$ when reading any trace on a path from $\hu$ to $\hv$. Moreover, for each loop in $\hH$, it captures whether $\bMM$ accepts the corresponding infinite trace.\footnote{A procedure to construct $\bMM$-behaviours is provided in~\cite[A.1]{techRep}.} 

The following theorem 
shows that checking $\bMM$-equivalence boils down to checking isomorphisms on $\bMM$-behaviours with views.

\begin{reptheorem}{lem:meq-mbhv}
For all $(\hH,U),(\hH',U') \in \hHHV$, we have
\begin{align*}
&	
(\eta,\mu) \vdash (\hH,U) \equiv_\bMM (\hH',U') \\ 
\qiff & 
(\eta \uplus \{(\star,\star)\},\mu) \vdash (\bMM(\hH),U \uplus \{\star\}) \cong (\bMM(\hH'),U' \uplus \{\star\}).
\end{align*}
\end{reptheorem}
For our example in \Cref{fig:meq-ex}, the black portion of the $\bMM$-behaviours in \Cref{fig:meq-ex-m} is exposed, showing that $(\hH,\emptyset)$ and $(\hJ,\emptyset)$ are indeed $\bMM$-equivalent. 

We use the above result to reason about the number of equivalence classes that can be constructed during grammar recoloring.
\begin{reptheorem}{thm:amt-meq}
For every HRG $\gGG$, we have
\[
|\eqclass{(\lLL_\gA(\gGG),\emptyset)}_{\equiv\bMM}| ~+~ |\eqclass{(\gCC_\gA(\gGG),\emptyset)}_{\equiv\bMM}|
	\quad=\quad 
	\bigO\big(2^{(\arity(\gA)^2 \times 2^{|\bQ|^2})}\big)
.
\]
Furthermore, for the recolored HRG $\gGG_\bMM$ with maximum number of hyperedges in rules $|\hE|$ and maximum arity of nonterminals $|\hI|$, we have
\[
	|\gGG_\bMM| 
	\quad=\quad
	\bigO\big(|\gP| \times |\gN|^{|\hE|} \times 2^{(|\hE| \times |\hI|^2 \times 2^{|\bQ|^2})}\big).
\]
\end{reptheorem}

\begin{figure}[t]
\centering	
\tikzset{
diagonal fill/.style 2 args={fill=#2, path picture={
\fill[#1, sharp corners] (path picture bounding box.south west) -|
                         (path picture bounding box.north east) -- cycle;}},
reversed diagonal fill/.style 2 args={fill=#2, path picture={
\fill[#1, sharp corners] (path picture bounding box.north west) |- 
                         (path picture bounding box.south east) -- cycle;}}
}
\begin{subfigure}[b]{.7\linewidth}
\centering	
\scalebox{0.65}{\begin{tikzpicture}[
node distance = 7mm and 7mm,
 Arr/.style = {thick,shorten >=8pt,shorten <=8pt},
 E/.style = {draw, thick, inner sep=0mm,
     minimum height=10mm, minimum width=26mm, },
 F/.style = {E, minimum height=7mm, minimum width=16mm},
 I/.style = {circle,fill=black, inner sep=0mm, minimum size=2mm},
 X/.style = {inner sep=0mm, minimum size=0mm},
 M/.style = {draw,inner sep=0mm, minimum size=1.5mm},
 V/.style = {circle,draw, inner sep=0mm, minimum size=2mm},
 VR/.style = {circle,draw,fill=red!40, inner sep=0mm, minimum size=2mm},
DL/.style = {thick, densely dotted, Latex-Latex,
             shorten >=#1, shorten <=#1},
DL/.default = 4mm,
L/.style = {-{Classical TikZ Rightarrow[width=4pt]}},
BL/.style = {->,thick, looseness=4},
IBL/.style = {L,thick, bend left},
IBR/.style = {L,thick, bend right},
ILL/.style = {L,thick, loop left, looseness=14},
ILR/.style = {L,thick, loop right, looseness=14},
ILU/.style = {L,thick, loop above, looseness=14},
T/.style = {opacity=0.2}]
 \newcommand{\nodeof}[2][]{\node(#2)[E,#1]{};\node[at=(#2), xshift=-10mm]{#2};\draw[thick,dashed]($(#2.north)+(-7mm,0)$)--($(#2.south)+(-7mm,0)$);\draw[thick,dashed]($(#2.north)+(3mm,0)$)--($(#2.south)+(3mm,0)$);\node(#2L)[at=(#2), xshift=-2mm,yshift=.5mm]{};\node(#2R)[at=(#2),xshift=+8mm,yshift=.5mm]{}}
    \def\xb{0}\def\yb{0}
    \draw[->,Arr] (\xb+1,\yb) to[out=-60,in=-120,distance=15] (\xb+2,\yb);
    \draw[->,Arr] (\xb+2,\yb) to[out=+120,in=+60,distance=15] (\xb+1,\yb);
    \draw[Arr]    (\xb+2,\yb) -- (\xb+3,\yb) node[pos=0.5, above, yshift=-1.5pt] {$\scriptstyle \mathbf 1$};
    \draw[Arr]   (\xb+5,\yb) -- (\xb+6,\yb) node[pos=0.5, above, yshift=-1.5pt] {$\scriptstyle \mathbf 2$};
    \draw (\xb+1,\yb) node[ver,inp] {$\mathbf{1}$};
    \draw (\xb+2,\yb) node[ver,circle,col=red] {};
    \draw (\xb+6,\yb) node[ver,inp] {$\mathbf{2}$};
    \draw (\xb-0.5,\yb) node {$:=$};
    \draw (\xb+0.25,\yb) node {$\gR_2$:};
 \node at (\xb-2.70,\yb+0.75) {$\gs$};
 \node at (\xb-1.70,\yb+0.75) {$\gx$};
 \node at (\xb+3.80,\yb+0.75) {$\gr$};
 \node at (\xb+4.80,\yb+0.75) {$\gy$};
 \nodeof[at={(\xb-2.5,\yb)}]{A};
 \node (AL1) [I, at=(AL), xshift=-3mm] {};
 \node (AL2) [I, at=(AL), xshift=3mm] {};
 \draw[ILU] (AL1) to (AL1);
 \draw[ILU] (AL2) to (AL2);
 \draw[IBL] (AL1) to (AL2);
 \draw[IBL] (AL2) to (AL1);
 \node (AR0) [M, at=(AR)] {};
 \node (AR1) [V, at=(AR), xshift=-3mm] {};
 \node (AR2) [V, at=(AR), xshift=3mm] {};
 \node (AR3) [X, at=(AR), yshift=-4mm] {$\star$};
 \draw[ILU] (AR1) to (AR1);
 \draw[ILU, blue] (AR2) to (AR2);
 \draw[IBL, blue] (AR2) to (AR3);
 \draw[IBR] (AR1) to (AR3);
 \draw (AR1) -- (AR0) -- (AR2);
 \nodeof[at={(\xb+4,\yb)}]{A};
 \node (AL1) [I, at=(AL), xshift=-3mm] {};
 \node (AL2) [I, at=(AL), xshift=3mm] {};
 \draw[IBL] (AL1) to (AL2);
 \draw[IBL] (AL2) to (AL1);
 \node (AR1) [V, at=(AR), xshift=-3mm] {};
 \node (AR0) [M, at=(AR)] {};
 \node (AR2) [V, at=(AR), xshift=3mm] {};
 \node (AR3) [X, at=(AR), yshift=-4mm] {$\star$};
 \draw[ILU] (AR1) to (AR1);
 \draw[ILU, blue] (AR2) to (AR2);
 \draw[IBL, blue] (AR2) to (AR3);
 \draw[IBR] (AR1) to (AR3);
 \draw (AR1) -- (AR0) -- (AR2);
\end{tikzpicture}}
\caption{\new{A rule of the refined HRG obtained from \Cref{def:refine}.}}
\label{fig:ex-refine}
\end{subfigure}
\begin{subfigure}[b]{.25\linewidth}
\centering	
\scalebox{0.65}{\begin{tikzpicture}[
node distance = 7mm and 7mm,
 Arr/.style = {thick,shorten >=8pt,shorten <=8pt},
 E/.style = {draw, thick, inner sep=0mm,
     minimum height=10mm, minimum width=26mm, },
 F/.style = {E, minimum height=7mm, minimum width=16mm},
 I/.style = {circle,fill=black, inner sep=0mm, minimum size=2mm},
 X/.style = {inner sep=0mm, minimum size=0mm},
 M/.style = {draw,inner sep=0mm, minimum size=1.5mm},
 V/.style = {circle,draw, inner sep=0mm, minimum size=2mm},
 VR/.style = {circle,draw,fill=red!40, inner sep=0mm, minimum size=2mm},
DL/.style = {thick, densely dotted, Latex-Latex,
             shorten >=#1, shorten <=#1},
DL/.default = 4mm,
L/.style = {-{Classical TikZ Rightarrow[width=4pt]}},
BL/.style = {->,thick, looseness=4},
IBL/.style = {L,thick, bend left},
IBR/.style = {L,thick, bend right},
ILL/.style = {L,thick, loop left, looseness=14},
ILR/.style = {L,thick, loop right, looseness=14},
ILU/.style = {L,thick, loop above, looseness=14},
T/.style = {opacity=0.2}]
 \newcommand{\nodeof}[2][]{\node(#2)[E,#1]{};\node[at=(#2), xshift=-10mm]{#2};\draw[thick,dashed]($(#2.north)+(-7mm,0)$)--($(#2.south)+(-7mm,0)$);\draw[thick,dashed]($(#2.north)+(3mm,0)$)--($(#2.south)+(3mm,0)$);\node(#2L)[at=(#2), xshift=-2mm,yshift=.5mm]{};\node(#2R)[at=(#2),xshift=+8mm,yshift=.5mm]{}}

    \def\xb{0}\def\yb{-3}
    \draw[->,Arr] (\xb+1,\yb) to[out=-60,in=-120,distance=15] (\xb+2,\yb);
    \draw[->,Arr] (\xb+2,\yb) to[out=+120,in=+60,distance=15] (\xb+1,\yb);
    \draw[->,Arr] (\xb+2,\yb) to[out=-60,in=-120,distance=15] (\xb+3,\yb);
    \draw[->,Arr] (\xb+3,\yb) to[out=+120,in=+60,distance=15] (\xb+2,\yb);
    \draw (\xb+1,\yb) node[ver,circle,thick] {};
    \draw (\xb+2,\yb) node[ver,circle, fill=red!40] {};
    \draw (\xb+3,\yb) node[ver,circle,thick] {};
    \draw (\xb+2,\yb-1) node {\scalebox{2}{$\star$}};
    \draw[->,Arr] (\xb+1,\yb) to[out=210,in=150,distance=30] (\xb+1,\yb);
    \draw[->,Arr,blue] (\xb+3,\yb) to[out=390,in=330,distance=30] (\xb+3,\yb);
    \draw[->,Arr] (\xb+1,\yb) to[out=-90,in=180,distance=15] (\xb+2,\yb-1);
    \draw[->,Arr,blue] (\xb+3,\yb) to[out=-90,in=0,distance=15] (\xb+2,\yb-1);
\end{tikzpicture}}
\caption{\new{$\repl{\gx}{\repl{(\gR_2,\cmU_{\gR_2})}{\gr}}$}}
\label{fig:ex-xrs}
\end{subfigure}
\caption{\new{Refining and recoloring an HRG rule from \Cref{fig:grammar-dll} with respect to $F \cblues$.}}
\label{fig:grammar-recol_ex}
\end{figure}
%
\noindent
To conclude this section, we illustrate how refinement and recoloring are applied to the HRG in \Cref{fig:grammar-dll} and the $\bMM$-equivalence obtained from $F \cblues$.
\Cref{fig:ex-refine} shows one out of 11 rules of the refined HRG; the other rules are found in \cite[A.2]{techRep}.
 The rule's nonterminals have been augmented by two $\bMM$-equivalence classes, representing (1) the behaviour of all graphs in the language generated by the nonterminal, and (2) the behaviour of all contexts generated by the nonterminal.
Similarly to \Cref{fig:meq-ex-m}, we illustrate the portion of $\bMM$-behaviours exposed by the empty view.
For simplicity, we represent paths that (do not) pass through a blue node by blue (black) arrows.
The rule in \Cref{fig:ex-refine} satisfies the constraints of \Cref{def:refine}, because $\gs = \repl{(\gR_2, \emptyset)}{\gr}$ and $\gy = \repl{\gx}{(\gR_2, \emptyset)}$.
After computing the refined HRG, we apply \Cref{def:gram-recol} by computing $\repl{\gx}{\repl{(\gR,\cmU_{\gR})}{\gr}}$ for each rule.
\Cref{fig:ex-xrs} illustrates the result for the rule in \Cref{fig:ex-refine}. Since there exists a trace that starts in the node added by $\gR_2$ (red) and takes only black edges, the node does not satisfy $F \cblues$. Hence, we do not add a new color to it.

\section{The recoloring algorithm for LTL and CTL*} \label{sec:impl}
\SetKwFunction{AddColor}{add}
\SetKwFunction{AddColorB}{$\textttn{add}_\Psi$}
\SetKwFunction{AddColorN}{$\textttn{add}_{\neg\Phi_1}$}
\SetKwFunction{AddColorA}{$\textttn{add}_{\Phi_1 \land \Phi_2}$}
\SetKwFunction{Minimize}{Minimize}
\SetKwFunction{Buchi}{B\"uchi}
\SetKwFunction{CheckC}{$\textttn{Recolor}_{\text{CTL*}}$}
\SetKwFunction{CheckL}{$\textttn{Recolor}_{\text{LTL}}$}
\SetKwFunction{AnnotateA}{$\textttn{Annotate}_1$}
\SetKwFunction{AnnotateB}{$\textttn{Annotate}_2$}

\SetKw{Closure}{compute closure of}
\SetKw{Fun}{Fun}
\SetKw{Returns}{return}

\begin{figure}[t]
\centering
\include{fig_code.tex}
\vspace{-1.5em}
\caption{Algorithms for HRG recoloring wrt. LTL and CTL* formulae.}
\label{algo:ctl_algo}
\end{figure}

We now present algorithms for recoloring HRGs with respect to formulae in \LTL and \CTLS (cf.~\cite{baierKatoen2008}) based on $\bMM$-behaviours.
These algorithms immediately yield decision procedures for \LTL and \CTLS model checking of families.
\Cref{algo:ctl_algo} summarizes our algorithms. 
We denote by $\texttt{labels}(\gP)$ an auxiliary function that returns all nonterminals that appear in the set of production rules $\gP$.

The algorithm \CheckL takes as input an HRG $\gGG$ and an \LTL formula $\varphi$. It then constructs the corresponding recolored \HRG as described in \Cref{sec:meq}.
The algorithm first converts $\varphi$ into a B\"uchi automaton $\bMM$.
It then invokes the procedure \AnnotateA, which annotates each nonterminal of the grammar with the first of the two equivalence classes, i.e. equivalence classes representing the graphs plugged into the hyperedges of production rules.
To this end, \AnnotateA iteratively computes production rules enriched by equivalence classes by applying \HRG rules to those obtained in previous iterations (starting with an empty set of rules). Intuitively, this process corresponds to a bottom-up computation of the equivalence classes of progressively larger elements of the \HRG's language.

After \AnnotateA has reached a fixed point, algorithm \CheckL calls the procedure \AnnotateB which assigns to each nonterminal its second equivalence class, i.e. the context in which rules are applied. As before, new equivalence classes are derived by repeatedly applying grammar rules to the previously obtained ones, until a fixed point is reached. Dually to \AnnotateA, this process intuitively corresponds to a top-down computation of the equivalence classes of progressively larger \emph{contexts}, using the equivalence classes computed by \AnnotateA to fill all but one of the holes. Since the $\bMM$-equivalence relation only allows for finitely many equivalence classes, the procedure is guaranteed to terminate, evaluating to the refined grammar $\gGG_\bMM$ from \Cref{def:refine}. 

Finally, \CheckL applies the function \AddColorB to color all relevant nodes with $\varphi$. Specifically, \AddColorB takes as input a function $f$ that maps rule-node pairs $(\gR,v)$ to Boolean values, and colors $v$ with $\varphi$ iff $f(\gR,v)$ evaluates to true. 
The considered condition is the same as in \Cref{def:gram-recol}, see~\Cref{algo:ctl_algo}.

The algorithm \CheckC in \Cref{algo:ctl_algo} (bottom right) shows how to extend our LTL recoloring algorithm to \CTLS, inspired by a standard algorithm for turning LTL model checkers into \CTLS model checkers~\cite[p. 427]{baierKatoen2008}.
Given an \HRG and a \CTLS formula $\Psi$, \CheckC reasons about \CTLS properties recursively.
Whenever we need to recolor a formula of the form $\forall\Psi$ containing CTL* subformulae $\Phi_1, \dots, \Phi_n$, we first recolor the \HRG according to each property $\Phi_i$, and then check $\forall\Psi$ as if each $\Phi_i$ were an atomic property (i.e. by reasoning directly on the colors). If the CTL* formula under consideration is instead a conjunction or a negation of CTL* subformulae, the grammar is first recolored according to those subformulae, and then each rule is recolored locally (e.g. a node is colored $\neg \alpha$ iff it is not colored $\alpha$). 
To reduce state space explosion, we minimize\footnote{This is an optional step. We discuss implementation details in \Cref{sec:eval}.} the \HRG after every invocation of \CheckL, since the annotations on nonterminals are no longer needed at that point. 

As discussed in \Cref{sec:problem}, once we have computed a recolored \HRG, model checking \LTL or \CTLS on families is straightforward. The same holds for the qualitative fragment of probabilistic computation tree logic (qPCTL) because, for finite-state models, qPCTL can be embedded into CTL~\cite[p. 791]{baierKatoen2008}.

\section{Implementation and Evaluation} \label{sec:eval}

\def\nhrgs{3\xspace}
\def\nprops{11\xspace}

We implemented a prototypical tool in Haskell (ca. 2k LOC) that covers our recoloring algorithms from \Cref{sec:impl} and \LTL, \CTLS, and qPCTL model checkers for families.
To translate LTL formulae into B\"uchi automata, we use LTL2BA~\cite{ltl2ba}.
Furthermore, our tool applies the following optimizations:
\begin{enumerate*}[label=(\roman*)]
	\item To minimize \HRGs after recoloring, we first translate them into tree automata accepting the grammar's derivation trees and then reduce their size using the VATA library~\cite{VATA,MINOTAUT} with the MinOTAut module. We also implemented a more expensive standard minimization algorithm for tree automata~\cite[p. 37]{TATA}.
	\item We prune the recolored \HRG whenever colors are known to be irrelevant for subsequent recolorings.
	\item The LTL recoloring algorithm is applied to subformulae in parallel and merges the results afterward. This optimization both enables concurrent execution and yields smaller grammars, because fewer colors are required.
\end{enumerate*}

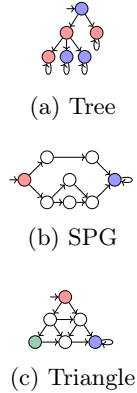
\begin{figure}[t]
\centering
\begin{subfigure}[c]{0.8\linewidth}
\centering
\scalebox{0.6}{\begin{tabular}{|l|l|l|l|l|l|l|l|l|l|}
\hline
\textbf{Case} & \textbf{Time (s)} & \textbf{Sat} & \textbf{Fal} & \textbf{R. \#\gN} & \textbf{R. \#\gP} & \textbf{M. \#\gN} & \textbf{M. \#\gP}  & \textbf{\#\hE} & \textbf{\#\hI}  \\
\hline\hline\hline\multicolumn{6}{|l|}{\textit{IPv4 zeroconf}} &2 & 3 & 1 & 3\\\hline
$(\forall F G\ \cblues) \lor (\forall G\ (\exists F\ \creds)))$ & 0.67 & $\infty$ & $0$ & 5 & 8 & 2 & 3 & \multicolumn{2}{l}{} \\\cline{1-8}
$X_{>0}(X_{>0}(X_{>0}\ \cblues))$ & 0.61 & \emph{fin} & $\infty$ & 12 & 18 & 5 & 9 & \multicolumn{2}{l}{} \\\cline{1-8}
$G_{=1}(F_{>0}(X_{=1}(\creds \lor \cblues) \land X_{>0}(\neg \cblues \land F_{>0}\ \cblues))$ & 1.04 & $0$ & $\infty$ & 5 & 8 & 2 & 3 & \multicolumn{2}{l}{} \\\cline{1-8}
\hline\hline\hline\multicolumn{6}{|l|}{\textit{Trees of arbitrary arity}} &3 & 8 & 2 & 1\\\hline
$\neg \exists (\cblues\ U\ \neg\cblues)$ & 0.59 & $\infty$ & $\infty$ & 17 & 140 & 5 & 22 & \multicolumn{2}{l}{} \\\cline{1-8}
$\forall G(\creds \lor \neg \cblues)$ & 0.42 & $\infty$ & $\infty$ & 9 & 74 & 5 & 22 & \multicolumn{2}{l}{} \\\cline{1-8}
$\neg(\exists F \cblues) \land \forall G(\creds \lor \neg \cblues) \land \forall G(\neg \creds)$ & 0.91 & $0$ & $\infty$ & 9 & 74 & 5 & 22 & \multicolumn{2}{l}{} \\\cline{1-8}
\hline\hline\hline\multicolumn{6}{|l|}{\textit{Series parallel graphs}} &2 & 4 & 2 & 2\\\hline
$\neg \exists (\cblues\ U\ \neg\cblues)$ & 0.46 & $0$ & $\infty$ & 16 & 98 & 2 & 4 & \multicolumn{2}{l}{} \\\cline{1-8}
$\forall G(\creds \lor \neg \cblues)$ & 0.45 & $0$ & $\infty$ & 16 & 98 & 2 & 4 & \multicolumn{2}{l}{} \\\cline{1-8}
$\neg(\exists F \cblues) \land \forall G(\creds \lor \neg \cblues) \land \forall G(\neg \creds)$ & 0.83 & $0$ & $\infty$ & 16 & 98 & 2 & 4 & \multicolumn{2}{l}{} \\\cline{1-8}
\hline\hline\hline\multicolumn{6}{|l|}{\textit{Sierpinski triangles}} &2 & 3 & 3 & 3\\\hline
$\exists F\ \cblues$ & 0.46 & $\infty$ & $0$ & 13 & 56 & 2 & 3 & \multicolumn{2}{l}{} \\\cline{1-8}
$\forall G(\exists F\ \cblues)$ & 0.60 & $\infty$ & $0$ & 13 & 56 & 2 & 3 & \multicolumn{2}{l}{} \\\cline{1-8}
$\exists F(\exists G\ \cblues)$ & 0.74 & $\infty$ & $0$ & 13 & 56 & 2 & 3 & \multicolumn{2}{l}{} \\\cline{1-8}
\hline\hline\hline\multicolumn{6}{|l|}{\textit{Bank attack trees}} &4 & 10 & 3 & 7\\\hline
$\exists F\ \creds$ & 17.04 & $\infty$ & $\infty$ & 78 & 1225 & 13 & 41 & \multicolumn{2}{l}{} \\\cline{1-8}
$\forall G\ \neg\cgreens$ & 9.80 & $\infty$ & $\infty$ & 30 & 421 & 9 & 26 & \multicolumn{2}{l}{} \\\cline{1-8}
$\exists X^{\le 2} (\creds \lor \cblues \lor \cgreens) \implies \exists F\ \cgreens$ & 298.45 & $\infty$ & \emph{fin} & 78 & 1225 & 13 & 41 & \multicolumn{2}{l}{} \\\cline{1-8}
\hline\hline\hline\multicolumn{6}{|l|}{\textit{Dining cryptographers}} &3 & 7 & 1 & 4\\\hline
$Correctness$ & 42.65 & $\infty$ & $0$ & 19 & 35 & 3 & 7 & \multicolumn{2}{l}{} \\\cline{1-8}
\end{tabular}}
\end{subfigure}
\begin{tabular}[c]{c}
\begin{subfigure}[c]{0.18\linewidth}
\centering
\scalebox{0.4}{\begin{tikzpicture}[thick]
	\def\ox{0}\def\oy{0}
	\def\sx{1}\def\sy{0.8}
    \node (V1) at (\ox+0*\sx,\oy+0*\sy) [ver1,circle,col=blue] {};
    \node (V2a) at (\ox-.5*\sx,\oy-1*\sy) [ver1,circle,col=red] {};
    \node (V2b) at (\ox+.5*\sx,\oy-1*\sy) [ver1,circle,col=red] {};
    \node (V3a) at (\ox-1.1*\sx,\oy-2*\sy) [ver1,circle,col=red] {};
    \node (V3b) at (\ox-.5*\sx,\oy-2*\sy) [ver1,circle,col=blue] {};
    \node (V3c) at (\ox+0.1*\sx,\oy-2*\sy) [ver1,circle,col=blue] {};
   	\draw[->] (\ox-.5,\oy+0) -- (V1);
	\draw[->] (V1) -- (V2a);
	\draw[->] (V1) -- (V2b);
	\draw[->] (V2a) -- (V3a);
	\draw[->] (V2a) -- (V3b);
	\draw[->] (V2a) -- (V3c);
	\draw[->, loop below	] (V2b) to (V2b);
	\draw[->, loop below] (V3a) to (V3a);
    \draw[->, loop below] (V3b) to (V3b);
    \draw[->, loop below	] (V3c) to (V3c);

\end{tikzpicture}}
\caption{Tree}
\label{fig:ex-tree}
\end{subfigure}
\\[3em]
\begin{subfigure}[c]{0.14\linewidth}
\centering
\scalebox{0.4}{\begin{tikzpicture}[thick]
	\def\ox{0}\def\oy{0}
	\def\sx{0.75}\def\sy{0.75}
	\node (V1) at (\ox+0*\sx,\oy+0*\sy) [ver1,circle,col=red] {};
    \node (V2) at (\ox+1*\sx,\oy+1*\sy) [ver1,circle] {};
    \node (V3) at (\ox+1*\sx,\oy-1*\sy) [ver1,circle] {};
    \node (V4) at (\ox+3*\sx,\oy+1*\sy) [ver1,circle] {};
    \node (V5) at (\ox+2*\sx,\oy+0*\sy) [ver1,circle] {};
    \node (V6) at (\ox+2*\sx,\oy-1*\sy) [ver1,circle] {};
    \node (V7) at (\ox+3*\sx,\oy-1*\sy) [ver1,circle] {};
    \node (V8) at (\ox+4*\sx,\oy+0*\sy) [ver1,circle,col=blue] {};
   	\draw[->] (\ox-.5,\oy+0) -- (V1);
   	\draw[->] (V1) -- (V2);
   	\draw[->] (V2) -- (V4);
   	\draw[->] (V4) -- (V8);
   	\draw[->] (V1) -- (V3);
   	\draw[->] (V3) -- (V5);
   	\draw[->] (V3) -- (V6);
   	\draw[->] (V5) -- (V7);
   	\draw[->] (V6) -- (V7);
   	\draw[->] (V7) -- (V8);
   	\draw[->, loop right] (V8) to (V8);
\end{tikzpicture}}
\caption{SPG}
\label{fig:ex-spg}
\end{subfigure}
\\[3em]
\begin{subfigure}[c]{0.14\linewidth}
\centering
\scalebox{0.4}{\begin{tikzpicture}[thick]
	\def\ox{0}\def\oy{0}
	\def\sx{1}\def\sy{1.5}
	\node (V1) at (\ox+0*\sx,\oy+0*\sy) [ver1,circle,col=red] {};
    \node (V2) at (\ox-1*\sx,\oy-1*\sy) [ver1,circle,col=ForestGreen] {};
    \node (V3) at (\ox+1*\sx,\oy-1*\sy) [ver1,circle,col=blue] {};
    \node (V12) at ($(V1)!0.5!(V2)$) [ver1,circle] {};
    \node (V13) at ($(V1)!0.5!(V3)$) [ver1,circle] {};
    \node (V23) at ($(V2)!0.5!(V3)$) [ver1,circle] {};
   	\draw[->] (\ox-.5,\oy+0) -- (V1);
	\draw[->] (V1) -- (V12);
    \draw[->] (V12) -- (V2);
    \draw[->] (V12) -- (V13);
    \draw[->] (V12) -- (V23);
    \draw[->] (V1) -- (V13);
    \draw[->] (V13) -- (V3);
    \draw[->] (V13) -- (V23);
    \draw[->] (V2) -- (V23);
    \draw[->] (V23) -- (V3);
    \draw[->, loop right] (V3) to (V3);
\end{tikzpicture}}
\caption{Triangle}
\label{fig:ex-sierp}
\end{subfigure}
\end{tabular}
\caption{An excerpt of our experiments. \textbf{Time} is the total runtime.
\textbf{Sat} is the number of models in which the property is satisfied ($0$, \emph{fin}itely many, or $\infty$),
\textbf{Fal} is the dual number of family members that do not satisfy property.
\textbf{R. \#\gP} and \textbf{R. \#\gN} are the number of rules and nonterminals of the refined \HRG.
\textbf{M. \#\gP} and \textbf{M. \#\gN} indicate the number of rules and nonterminals after minimization of the recolored grammar;  the original HRG's size is in the grammar's row.
\textbf{\#\hE} and \textbf{\#\hI} are the maximum number of hyperedges and abstract nodes per rule.}
\label{table:eval}
\end{figure}

We evaluated our tool on 56 model checking benchmarks for various \LTL, \CTLS and qPCTL formulae and four families inspired by the literature.
\Cref{table:eval} shows an excerpt of our experimental evaluation; the full table is found in \cite[A.3]{techRep}.
All experiments were performed on a MacBook Pro M4 Pro with 24 GB RAM.
We run experiments on the following families:
\begin{itemize}
\item Labeled transition systems representing Markov chain models of the \emph{IPv4 zeroconf} protocol~(cf.~\cite{DBLP:journals/fmsd/KwiatkowskaNPS06} and \cite[p. 752]{baierKatoen2008}) for any number $k \in \Nats$ of probes and some fixed probability $p> 0$.
\item \emph{Trees of arbitrary arity} and with arbitrary colors, e.g. \Cref{fig:ex-tree}, where properties are checked from the root node. This benchmark effectively turns our tool into a CTL* satisfiability checker for tree models, since extracting a witness from the recolored grammar only requires finding a derivation tree which uses a rule that has been appropriately recolored. 
\item \emph{Series parallel graphs}~\cite{DBLP:conf/gg/DrewesKH97,DBLP:journals/pacmpl/AlurSW23} of arbitrary size. See \Cref{fig:ex-spg}.
\item \emph{Sierpinski triangles}~\cite{plump2009graph} of arbitrary size, with the top-most, left-most and right-most nodes identified by a different color, see \Cref{fig:ex-sierp} for an example.
\item \emph{Robbery} is a variant of the SafeLock case study from \cite{DBLP:journals/tse/BeekLLV20}, which considers software product lines for security-related applications. This system models an attempt to rob a bank and open a safe. Multiple features represent possibilities to do so, e.g. guessing the safe's combination, blackmail, gathering information, or brute force. Our model considers different times for those approaches and admits failures, repetitions of attempts, and the potential of being spotted and jailed. Among the verified properties are: (1) one is never jailed; (2) there exists an approach to open the safe; (3) fast and dangerous attempts will always mean that one eventually ends up in jail.
\item We prove correctness of the \emph{Dining cryptographers} protocol \cite{Cha88} for an arbitrary number of participants.
\end{itemize}
\noindent
Overall, our prototype manages to verify properties for all six families, often in less than a second.
Unsurprisingly (see \Cref{thm:amt-meq}), performance does not scale well with the maximum number of hyperedges and external nodes.


\section{Related Work}\label{sec:related}
There is a vast body of literature on the verification of infinite-state systems based on, amongst others, pushdown automata (cf.~\cite{Song2011EfficientCM}), tree automata~\cite{Bouajjani2012AbstractR}, rewrite systems~\cite{Lding2002ModelCheckingIS,eker2004maude}, programs~\cite{Cook2015OnAO}, and hierarchical models~\cite{DBLP:journals/toplas/AlurY01,Lohrey2005ModelcheckingHS}.
This paper uses graph grammars for modeling and reasoning about \emph{infinite families of finite-state systems}.
We focus on two classes of related work: 
(1) model checking with graph grammars and (2) verification of families of finite-state models.

\emph{Reasoning about graph grammars.}
Burkart and Quemener~\cite{Burkart1996ModelcheckingOI} extended the state-labeling algorithm for CTL model checking to reason about infinite-state Kripke structures.
Apart from considering CTL instead of CTL*, they use HRGs to model a single infinite-state Kripke structure rather than an infinite family of finite-state systems.
Consequently, their HRGs consist of a single production rule with a single nonterminal hyperedge, which severely restricts expressiveness.

In principle, one can obtain a decision procedure for CTL* model checking of HRG families by encoding CTL* in monadic second order logic (MSO, cf.~\cite{DBLP:journals/corr/LaroussinieM14}) 
and then applying Courcelle's theorem~\cite{courcelle2012graph}, which states that it is decidable (with non-elementary complexity) whether all graphs generated by an HRG satisfy a given MSO formula over graphs.
By contrast, we developed and implemented a direct decision procedure based on the novel abstraction of $\bMM$-behaviours.
Our procedure yields a recolored HRG with the same production rules: The colors of every node in every rule are extended according to the satisfied CTL* formulae.
Hence, our algorithm naturally extends the classical state-labeling algorithm for CTL* model checking from finite-state systems to HRG families. 

Groove~\cite{Rensink2003TowardsMC,Rensink2008ExplicitSM} supports CTL model checking of graph languages by explicit enumeration.
The implemented algorithms are not decision procedures.

\emph{Families of finite-state models}
appear in various application areas.
	Software product lines (SPLs, cf.~\cite{DBLP:journals/csur/ThumAKSS14}) are families of models obtained by composing product features that may be switched on or off. 
Ensuring that all possible product configurations are safe thus amounts to reasoning about all family members.
A common approach is to assume that a system with all features is finite-state and then attempt to check all subsystems~\cite{DBLP:conf/icse/ClassenHSL11,DBLP:journals/fac/ChrszonDKB18}.
However, SPLs can describe infinite families, e.g. if features can be added to a system multiple times.
Pragmatic approaches thus fall back to testing a sample of all possible subsystems~\cite{DBLP:conf/se/KowalST15}.
This paper enables compositional CTL* model checking of infinite SPLs represented by graph grammars, which are used as a modeling formalism for SPLs~\cite{DBLP:journals/cera/DuJT02}.

For distributed systems, family verification is better known as parametric verification and has been extensively studied (cf.~\cite{DBLP:conf/stacs/Esparza14,DBLP:series/synthesis/2015Bloem}). 
While parametric verification has not been the focus of this work, some distributed systems have been modeled by HRGs~\cite{DBLP:conf/fmco/FerrariHLMT05} and can thus, in principle, be verified with our algorithm.

Finally, our algorithm enables both checking whether all or \emph{some} family member satisfies a CTL* formula. Hence, it can be used to check CTL* satisfiability with respect to models of interest: Given a formula and an HRG family, e.g. only tree-shaped or circular systems, does there exist a model of the formula that belongs to the family?
Such questions arise, for example, in bounded synthesis, which searches for a small system that satisfies an LTL formula (cf.~\cite{ehlers2012symbolic,DBLP:conf/stacs/KupfermanLVY11}).

\smallskip
\noindent
\emph{Conclusion.}
We studied model checking algorithms for graph grammars modeling infinite families of finite-state labeled transition systems against temporal properties written in CTL*.
Our algorithms verify whether all, some, or (in)finitely many family members satisfy a CTL* property.
We implemented our approach and presented on early experiments.
Future work includes improving grammar minimization and extending our approach to quantitative properties, e.g. simple fragments of PCTL, of families of Markov models.

\smallskip
\noindent
\emph{Acknowledgements.} We thank Alberto Lluch Lafuente for pointing us to the SafeLock case study.

\section*{Data-Availability Statement}

Our prototype software as well as all data and instructions required to reproduce our experiments are available online:
\begin{center}
	\url{https://zenodo.org/records/18175332}
\end{center}

\bibliographystyle{splncs04}
\bibliography{refs.bib}

\newpage
\appendix
\crefalias{section}{appendix}
\section{Additional Material}
\subsection{Construction of $\bMM$-behaviours} \label[appendix]{appendix:mbhv}
The components $\rightsquigarrow^*$ and $\rightsquigarrow^\omega$ of the $\bMM$-behaviour of a hypergraph $\hH$ can be constructed by taking the closure of the following rules:
$$
	\infer{\hu \rightsquigarrow^*_f \hv}{\stackanchor{$\hu \rightarrow_{\hH} \hv$}{$f = \{ (\bp,\bq,\boolb) \mid \bD^*_{\bMM}(\bp,\hC_{\hH}(\hu\hv \vert_{\hV_{\hH}}),\bq,\boolb) \}$}}
\qquad
	\infer{\hu \rightsquigarrow^*_g w}{\stackanchor{$\hu \rightsquigarrow^*_g \hv \quad \hv \in \hV_{\hH} \quad\hv \rightarrow_{\hH} w$}{$f = \left\{ (\bp,\bq,\boolb_1 \lor \boolb_2) \middle\bracevert \substack{\ensuremath{\bD_\bMM^*(\bp, \hC_\hH(\hw \vert_{\hV_\hH}), \bp',b_1) \\\land (\bp',\bq,\boolb_2) \in g}} \right\}$}}
$$$$
	\infer{(\bp,\bq,\boolb_1 \lor \boolb_2) \in \texttt{clos}(f)}{(\bp,\bp',\boolb_1) \in \texttt{clos}(f) \qquad (\bp',\bq,\boolb_2) \in \texttt{clos}(f)}
$$$$
	\infer{\hu \rightsquigarrow^\omega_f \star}{\stackanchor	{$\hu \rightsquigarrow^*_{g} \hv \quad \hv \rightarrow_\hH \hu \quad \hu,\hv \in \hV_\hH$}{$f = \{ \bp \mid (\bp,\bp,\boolT) \in \texttt{clos}(g) \}$}}
\qquad
	\infer{\hu \rightsquigarrow^\omega_f \star}{\stackanchor{$\hu \rightsquigarrow^*_{g} \hv \quad \hv \rightarrow_\hH \hw \quad \hw \rightsquigarrow^\omega_{h} \star$}{$f = \{ \bp \mid (\bp,\bq,\boolb) \in g \land \bq \in h \}$}}
$$	
As the domains we're dealing with are finite, the closure will reach a fixpoint in finite steps.

\newpage
\subsection{Full Recolored Example HRG} \label[appendix]{appendix:ex-rcoloring}

\begin{figure*}
\centering
\includegraphics[width=0.4\linewidth]{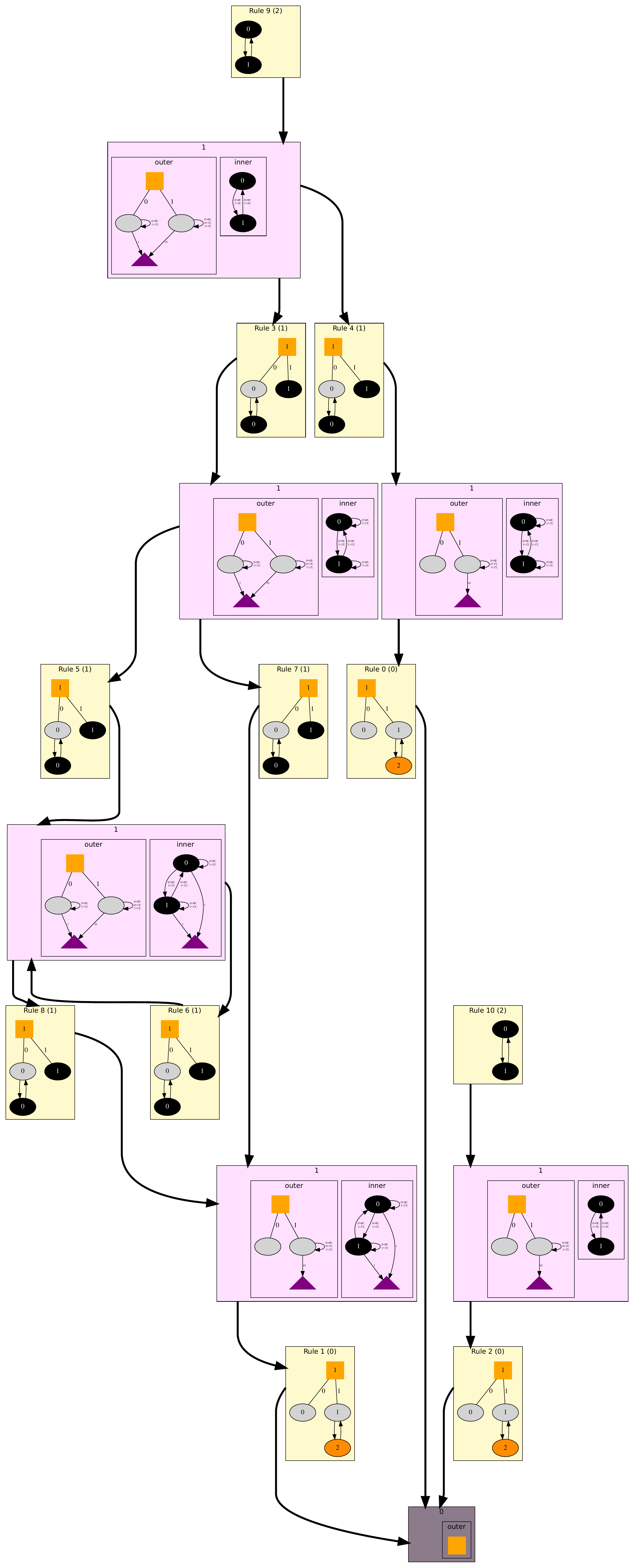}
\caption{Refined grammar for property $F \cblues$ on the grammar from \Cref{fig:grammar-dll}.}
\label{fig:tool-out}
\end{figure*}

\noindent 
\Cref{fig:tool-out} shows all rules of the refined grammar from the example at the end of \Cref{sec:meq}.
The figure has been generated by our prototype tool.
The figure arranges the grammar as a tree-automaton where nonterminals are states (pink boxes), and rules are transitions (yellow boxes). The rule from \Cref{fig:ex-refine} corresponds to \emph{"Rule 3"} and connects the state corresponding to the nonterminal of the hyperedge, to the nonterminal of the rule. Each state contains two graphs corresponding to the $\bMM$-behaviours $\gx$ and $\gs$, with the small purple node corresponding to the special node $\star$.
\newpage
\subsection{Full Experiments}\label[appendix]{sec:full-experiments}
\begin{figure}
\centering
\scalebox{0.5}{\begin{tabular}{|l|l|l|l|l|l|l|l|l|l|}
\hline
\textbf{Case} & \textbf{Time (s)} & \textbf{Sat} & \textbf{Fal} & \textbf{R. \#\gN} & \textbf{R. \#\gP} & \textbf{M. \#\gN} & \textbf{M. \#\gP}  & \textbf{\#\hE} & \textbf{\#\hI}  \\
\hline\hline\hline\multicolumn{6}{|l|}{\textit{IPv4 zeroconf}} &2 & 3 & 1 & 3\\\hline
$\exists X\ \creds$ & 3.55 & $0$ & $\infty$ & 7 & 11 & 2 & 3 & \multicolumn{2}{l}{} \\\cline{1-8}
$\forall G(\exists X\ \creds)$ & 0.47 & $0$ & $\infty$ & 5 & 8 & 2 & 3 & \multicolumn{2}{l}{} \\\cline{1-8}
$\exists X\ \cblues$ & 0.47 & \emph{fin} & $\infty$ & 7 & 11 & 3 & 5 & \multicolumn{2}{l}{} \\\cline{1-8}
$\forall G(\exists X\ \cblues)$ & 0.46 & $0$ & $\infty$ & 7 & 11 & 2 & 3 & \multicolumn{2}{l}{} \\\cline{1-8}
$\exists F G\ \cblues$ & 0.42 & $\infty$ & $0$ & 5 & 8 & 2 & 3 & \multicolumn{2}{l}{} \\\cline{1-8}
$\forall F G\ \cblues$ & 0.35 & $0$ & $\infty$ & 5 & 8 & 2 & 3 & \multicolumn{2}{l}{} \\\cline{1-8}
$\exists F\ \creds$ & 0.39 & $\infty$ & $0$ & 5 & 8 & 2 & 3 & \multicolumn{2}{l}{} \\\cline{1-8}
$\forall G(\exists F\ \creds)$ & 0.43 & $\infty$ & $0$ & 5 & 8 & 2 & 3 & \multicolumn{2}{l}{} \\\cline{1-8}
$(\forall F G\ \cblues) \lor (\forall G\ (\exists F\ \creds)))$ & 0.67 & $\infty$ & $0$ & 5 & 8 & 2 & 3 & \multicolumn{2}{l}{} \\\cline{1-8}
$G_{>0}\ \cblues$ & 0.50 & $0$ & $\infty$ & 5 & 8 & 2 & 3 & \multicolumn{2}{l}{} \\\cline{1-8}
$F_{>0}(G_{>0}\ \cblues)$ & 0.51 & $0$ & $\infty$ & 5 & 8 & 2 & 3 & \multicolumn{2}{l}{} \\\cline{1-8}
$F_{>0}\ \creds$ & 0.44 & $\infty$ & $0$ & 5 & 8 & 2 & 3 & \multicolumn{2}{l}{} \\\cline{1-8}
$G_{=1}(F_{>0}\ \creds)$ & 0.47 & $\infty$ & $0$ & 5 & 8 & 2 & 3 & \multicolumn{2}{l}{} \\\cline{1-8}
$G_{>0}\ \neg\creds$ & 0.62 & $0$ & $\infty$ & 5 & 8 & 2 & 3 & \multicolumn{2}{l}{} \\\cline{1-8}
$F_{>0}\ \cblues$ & 0.43 & $\infty$ & $0$ & 5 & 8 & 2 & 3 & \multicolumn{2}{l}{} \\\cline{1-8}
$G_{>0}(F_{>0}\ \cblues)$ & 0.61 & $\infty$ & $0$ & 5 & 8 & 2 & 3 & \multicolumn{2}{l}{} \\\cline{1-8}
$X_{>0}\ \cblues$ & 0.37 & \emph{fin} & $\infty$ & 7 & 11 & 3 & 5 & \multicolumn{2}{l}{} \\\cline{1-8}
$X_{>0}(X_{>0}\ \cblues)$ & 0.49 & \emph{fin} & $\infty$ & 9 & 14 & 4 & 7 & \multicolumn{2}{l}{} \\\cline{1-8}
$X_{>0}(X_{>0}(X_{>0}\ \cblues))$ & 0.61 & \emph{fin} & $\infty$ & 12 & 18 & 5 & 9 & \multicolumn{2}{l}{} \\\cline{1-8}
$X_{=1}(\creds \lor \cblues)$ & 0.47 & \emph{fin} & $\infty$ & 7 & 11 & 3 & 5 & \multicolumn{2}{l}{} \\\cline{1-8}
$X_{>0}(\neg \cblues \land F_{>0}\ \cblues)$ & 0.68 & $\infty$ & \emph{fin} & 7 & 11 & 3 & 5 & \multicolumn{2}{l}{} \\\cline{1-8}
$F_{>0}(X_{=1}(\creds \lor \cblues) \land X_{>0}(\neg \cblues \land F_{>0}\ \cblues)$ & 0.93 & $0$ & $\infty$ & 5 & 8 & 2 & 3 & \multicolumn{2}{l}{} \\\cline{1-8}
$G_{=1}(F_{>0}(X_{=1}(\creds \lor \cblues) \land X_{>0}(\neg \cblues \land F_{>0}\ \cblues))$ & 1.04 & $0$ & $\infty$ & 5 & 8 & 2 & 3 & \multicolumn{2}{l}{} \\\cline{1-8}
\hline\hline\hline\multicolumn{6}{|l|}{\textit{Trees of arbitrary arity}} &3 & 8 & 2 & 1\\\hline
$\neg \exists (\cblues\ U\ \neg\cblues)$ & 0.59 & $\infty$ & $\infty$ & 17 & 140 & 5 & 22 & \multicolumn{2}{l}{} \\\cline{1-8}
$\exists F\ \cblues$ & 0.48 & $\infty$ & $\infty$ & 17 & 140 & 5 & 22 & \multicolumn{2}{l}{} \\\cline{1-8}
$\exists X\ \neg \cblues$ & 2.83 & $\infty$ & $\infty$ & 65 & 2084 & 6 & 32 & \multicolumn{2}{l}{} \\\cline{1-8}
$\exists F\ \cblues \land \exists X\ \neg \cblues$ & 3.03 & $\infty$ & $\infty$ & 65 & 2084 & 6 & 32 & \multicolumn{2}{l}{} \\\cline{1-8}
$\neg \forall G(\exists F\ \cblues \land \exists X\ \neg \cblues)$ & 3.13 & $\infty$ & $0$ & 9 & 82 & 3 & 8 & \multicolumn{2}{l}{} \\\cline{1-8}
$\neg(\exists F\ \cblues)$ & 0.51 & $\infty$ & $\infty$ & 17 & 140 & 5 & 22 & \multicolumn{2}{l}{} \\\cline{1-8}
$\forall G(\creds \lor \neg \cblues)$ & 0.42 & $\infty$ & $\infty$ & 9 & 74 & 5 & 22 & \multicolumn{2}{l}{} \\\cline{1-8}
$\forall G(\neg \creds)$ & 0.37 & $\infty$ & $\infty$ & 9 & 74 & 5 & 22 & \multicolumn{2}{l}{} \\\cline{1-8}
$\neg(\exists F \cblues) \land \forall G(\creds \lor \neg \cblues) \land \forall G(\neg \creds)$ & 0.91 & $0$ & $\infty$ & 9 & 74 & 5 & 22 & \multicolumn{2}{l}{} \\\cline{1-8}
\hline\hline\hline\multicolumn{6}{|l|}{\textit{Series parallel graphs}} &2 & 4 & 2 & 2\\\hline
$\neg \exists (\cblues\ U\ \neg\cblues)$ & 0.46 & $0$ & $\infty$ & 16 & 98 & 2 & 4 & \multicolumn{2}{l}{} \\\cline{1-8}
$\exists F\ \cblues$ & 0.35 & $\infty$ & $0$ & 16 & 98 & 2 & 4 & \multicolumn{2}{l}{} \\\cline{1-8}
$\exists X\ \neg \cblues$ & 1.99 & $\infty$ & $\infty$ & 155 & 2185 & 4 & 12 & \multicolumn{2}{l}{} \\\cline{1-8}
$\exists F\ \cblues \land \exists X\ \neg \cblues$ & 2.19 & $\infty$ & $\infty$ & 155 & 2185 & 4 & 12 & \multicolumn{2}{l}{} \\\cline{1-8}
$\neg \forall G(\exists F\ \cblues \land \exists X\ \neg \cblues)$ & 2.52 & $\infty$ & $0$ & 20 & 144 & 2 & 4 & \multicolumn{2}{l}{} \\\cline{1-8}
$\neg(\exists F\ \cblues)$ & 0.38 & $0$ & $\infty$ & 16 & 98 & 2 & 4 & \multicolumn{2}{l}{} \\\cline{1-8}
$\forall G(\creds \lor \neg \cblues)$ & 0.45 & $0$ & $\infty$ & 16 & 98 & 2 & 4 & \multicolumn{2}{l}{} \\\cline{1-8}
$\forall G(\neg \creds)$ & 0.36 & $0$ & $\infty$ & 16 & 98 & 2 & 4 & \multicolumn{2}{l}{} \\\cline{1-8}
$\neg(\exists F \cblues) \land \forall G(\creds \lor \neg \cblues) \land \forall G(\neg \creds)$ & 0.83 & $0$ & $\infty$ & 16 & 98 & 2 & 4 & \multicolumn{2}{l}{} \\\cline{1-8}
\hline\hline\hline\multicolumn{6}{|l|}{\textit{Sierpinski triangles}} &2 & 3 & 3 & 3\\\hline
$\exists F\ \cblues$ & 0.46 & $\infty$ & $0$ & 13 & 56 & 2 & 3 & \multicolumn{2}{l}{} \\\cline{1-8}
$\forall G(\exists F\ \cblues)$ & 0.60 & $\infty$ & $0$ & 13 & 56 & 2 & 3 & \multicolumn{2}{l}{} \\\cline{1-8}
$\exists F\ \cgreens$ & 0.77 & $\infty$ & $0$ & 19 & 83 & 5 & 9 & \multicolumn{2}{l}{} \\\cline{1-8}
$\forall G(\exists F\ \cgreens)$ & 1.47 & $0$ & $\infty$ & 19 & 83 & 2 & 3 & \multicolumn{2}{l}{} \\\cline{1-8}
$\exists G\ \cblues$ & 0.62 & $0$ & $\infty$ & 13 & 56 & 2 & 3 & \multicolumn{2}{l}{} \\\cline{1-8}
$\exists F(\exists G\ \cblues)$ & 0.74 & $\infty$ & $0$ & 13 & 56 & 2 & 3 & \multicolumn{2}{l}{} \\\cline{1-8}
$\exists X(\cblues)$ & 28.11 & \emph{fin} & $\infty$ & 61 & 2176 & 4 & 7 & \multicolumn{2}{l}{} \\\cline{1-8}
$\exists X(\cgreens)$ & 31.75 & \emph{fin} & $\infty$ & 55 & 1959 & 6 & 14 & \multicolumn{2}{l}{} \\\cline{1-8}
$\exists X(\cblues) \land \exists X(\cgreens)$ & 55.55 & \emph{fin} & $\infty$ & 55 & 1959 & 6 & 14 & \multicolumn{2}{l}{} \\\cline{1-8}
\hline\hline\hline\multicolumn{6}{|l|}{\textit{Bank attack trees}} &4 & 10 & 3 & 7\\\hline
$\exists F\ \creds$ & 17.04 & $\infty$ & $\infty$ & 78 & 1225 & 13 & 41 & \multicolumn{2}{l}{} \\\cline{1-8}
$\forall G\ \neg\cgreens$ & 9.80 & $\infty$ & $\infty$ & 30 & 421 & 9 & 26 & \multicolumn{2}{l}{} \\\cline{1-8}
$\exists X^{\le 2} (\creds \lor \cblues \lor \cgreens)$ & 308.57 & \emph{fin} & $\infty$ & 80 & 7082 & 7 & 16 & \multicolumn{2}{l}{} \\\cline{1-8}
$\exists F \cgreens$ & 10.01 & $\infty$ & $\infty$ & 30 & 421 & 9 & 26 & \multicolumn{2}{l}{} \\\cline{1-8}
$\exists X^{\le 2} (\creds \lor \cblues \lor \cgreens) \implies \exists F \cgreens$ & 298.45 & $\infty$ & \emph{fin} & 78 & 1225 & 13 & 41 & \multicolumn{2}{l}{} \\\cline{1-8}
\hline\hline\hline\multicolumn{6}{|l|}{\textit{Dining Cryptographers}} &3 & 7 & 1 & 4\\\hline
$Correctness$ & 42.65 & $\infty$ & $0$ & 19 & 35 & 3 & 7 & \multicolumn{2}{l}{} \\\cline{1-8}
\end{tabular}}
\caption{Full evaluation table.}
\end{figure}

\newpage
\section{Proofs}

We define a generalization of contexts $\gCC_\gA^\gB(\gGG)$ such that $\sem{\tT} = \hJ \in \gCC^\gB_\gA(\gGG)$ iff $\tT$ is a derivation tree of $\gGG$ with root label $\gB$ (not necessarily in $\gZ$) and with exactly one sub-tree equal to a leaf, which must be equal to $\gA$, redefining $\gCC_\gA(\gGG)$ as $\bigcup_{\gB \in \gZ_\gGG} \gCC_\gA^\gB(\gGG)$. We will use the following fact in many of the following proofs.
\begin{proposition}
All hypergraphs $\hH \in \lLL_\gA(\gGG)$ (for arbitrary $\gA$) can be constructed inductively, as follows:
$$
	\infer
		{\repl{\gR}{\fml\hH} \in \lLL_\gA(\gGG)}
		{
			\grule{\gR}{\fml\gB}{\gA} \in \gP_\gGG
			&
			\fml\hH(\he) \in \lLL_{\fml\gB(\he)}(\gGG)
		}
$$
Similarly, all contexts $\hH \in \gCC_\gA^\gB(\gGG)$ (for arbitrary $\gA$ and $\gB$) can be constructed inductively, as follows:
$$
	\infer{\handle\gA \in \gCC^\gA_{\gA}(\gGG)}{\gA \in \gN_\gGG} 
	\qquad 
	\infer
		{\repl{\hH}{\repl{\gR}{\fml\hJ\langle \he / \handle{\fml\gB(\he)} \rangle}} \in \gCC^\gC_{\fml\gB(\he)}(\gGG)}
		{
			\hH \in \gCC^\gC_\gA(\gGG) 
			&
			\grule{\gR}{\fml\gB}{\gA} \in \gP_\gGG
			&
			\fml\hJ(\he) \in \lLL_{\fml\gB(\he)}(\gGG)
		}
$$
\end{proposition}

\subsection{Proof of \Cref{lem:hg-decomposition}}

\begin{proposition}[\cite{DBLP:conf/gg/DrewesKH97}] \label{lem:assoc}
Let $\hH,\hJ\in \hHH$ and $\fml\hK = (\hK_\he \in \hHH)_{\he \in \hE_\hJ}$ with $|\hE_\hH| = 1$. Then,
$$
	\repl{\hH}{\repl{\hJ}{\fml\hK}} = \repl{\repl{\hH}{\hJ}}{\fml\hK}	
$$	
\end{proposition}

\repeatlemma{lem:hg-decomposition}

\begin{proof}
We prove, by induction, the stronger claim on languages on arbitrary nonterminals, and contexts on arbitrary nonterminals. Let $\grule{\gR}{\fml{\gB}}{\gA} \in \gP_{\gGG}$ be an arbitrary rule in the grammar, and let $\fml{\hH}$ such that $\fml\hH(\he) \in \lLL_{\fml\gB(\he)}(\gGG)$ be an assignment to hyperedges such that $\repl{\gR}{\fml\hH} \in \lLL_\gA(\gGG)$. We assume, as our inductive hypothesis, that our claim holds on all hypergraphs $\hH(\he)$.

To prove our claim, we have to show that it is possible to split, as described, the graph $\repl{\gR}{\fml{\hH}}$ exposing any $\hu \in \hV_{\repl{\gR}{\fml{\hH}}}$. Let $\hu$ be such an arbitrary node. Then, either $\hu \in \hV_{\gR}$, or $\hu \in \hV_{\hH_{\he'}}$ for some $\he' \in \hE_{\gR}$. In the first case, we have that
$$
	\repl{\gR}{\fml{\hH}} = \repl{\handle{\gA}}{\repl{\gR}{\fml{\hH}}} \quad \text{and} \quad \hu \in \hV_{\gR},
$$
as required. In the second case, from the inductive hypothesis we know that, for some $\hJ$ and $\fml\hK$, it holds that 
$$
	\hH_{\he'} = (\repl{\hJ}{\repl{\gR'}{\fml{\hK}}}) \quad \text{and} \quad \hu \in \hV_{\gR'},
$$
but then, letting $\fml{\hH}' = \fml\hH\langle \he / \handle\gA\rangle$ denote the family of HGs such that $\hH'_{\he} = \handle{\gA}$, and $ \forall \he' \ne \he.\; \hH'_{\he'}$ = $\hH_{\he'}$, we have that 
$$
	\repl{\gR}{\fml{\hH}} = \repl{\gR}{\repl{\fml{\hH}'}{\hH_{\he}}} = \repl{\gR}{\repl{\fml{\hH}'}{\repl{\hJ}{\repl{\gR'}{\fml{\hK}}}}} = \repl{(\repl{\gR}{\repl{\fml{\hH}'}{\hJ}})}{\repl{\gR'}{\fml{\hK}}}
$$
as required, via \Cref{lem:assoc}.
\end{proof}

\subsection{Proof of \Cref{thm:ref-ok}}

\repeattheorem{thm:ref-ok}

\begin{proof}
We show the two claims given by the theorem separately.

\textit{Claim 1.} We begin by proving by induction that for any $\hH \in \lLL_{(\gA,\gs,\gx)}(\gGG_\equiv)$ it holds that $\gs = \eqclass{(\emptyset, \hH)}_\equiv$. Let $\grule{\gR}{\fml\gB \zip \fml\gr \zip \fml\gy}{(\gA,\gs,\gx)} \in \gP_{\gGG_\equiv}$ be an arbitrary rule in the grammar, and let 
\begin{equation}
	\fml{\hH}(\he) \in \lLL_{(\fml\gB(\he),\fml\gr(\he),\fml\gy(\he))}(\gGG_\equiv) \quad\st\quad \fml\gr(\he) = \eqclass{(\fml\hH(\he), \emptyset)}_\equiv \tag{I.H.}
\end{equation}
be a family of hypergraphs granted by our inductive hypothesis. Then, from \Cref{def:congr} and from our inductive hypothesis, we have that 
$$
	\gs = \repl{(\gR,\emptyset)}{\fml\gr} = \eqclass{(\repl{\gR}{\fml\hH}, \emptyset)}_\equiv.
$$

\textit{Claim 2.} We prove by induction that for any $\hJ \in \gCC_{(\gA,\gs,\gx)}(\gGG_\equiv)$ it holds that $\gx = \eqclass{(\emptyset, \hJ)}_\equiv$. First, let $(\gA,\gs,\gx) \in \gZ_{\gGG_\equiv}$ be a starting symbol. Then, we have that $\handle{(\gA,\gs,\gx)} \in \gCC_{(\gA,\gs,\gx)}(\gGG_\equiv)$. Thus, by \Cref{def:congr}
$$
	\gx = \eqclass{(\handle{\gA}, \emptyset)}_\equiv = \eqclass{(\handle{(\gA,\gs,\gx)}, \emptyset)}_\equiv,
$$
as all handles of the same size are $\bMM$-equivalent (as they have no paths). Next, let $\grule{\gR}{\fml{\gB,\gr,\gy}}{(\gA,\gs,\gx)} \in \gP_{\gGG_\equiv}$, let
$$
	\hH \in \gCC_{(\gA,\gs,\gx)}(\gGG) \quad\st\quad \gx = \eqclass{(\hH, \emptyset)}_\equiv
$$ 
be a context granted by our inductive hypothesis and let 
$$
	\fml\hJ(\he) \in \lLL_{(\fml\gB(\he),\fml\gr(\he),\fml\gy(\he))}(\gGG_\equiv) \quad\st\quad \fml\gr(\he) = \eqclass{(\hJ(\he), \emptyset)}_\equiv
$$ 
be a family of hypergraphs, granted by the first part of the theorem. Then,
$$
	\fml\gy(\he) = \repl{\gx}{\repl{\gR}{\fml\gr\langle \he / \handle{\fml\gB(\he)} \rangle}} = \eqclass{(\repl{\hH}{\repl{\gR}{\fml\hJ\langle \he / \handle{\fml\gB(\he)} \rangle}}, \emptyset)}_\equiv.
$$
\end{proof}

\subsection{Proof of \Cref{lem:meq-cong}}
We define the set of finite and infinite concrete paths as
\[
\CPaths^*_{\hH}(\hu,\hv) ~=~ \left\{ \hu\seq\hpath\hv \vert_{\hV_{\hH}} \middle\bracevert \hu\seq\hpath\hv \in \Conn_{\hH} ~\tand~ \seq\hpath \in \hV_\hH^* \right\},
\]
and 
\[
\CPaths^\omega_{\hH}(\hu) ~=~ \left\{ \hu\seq\hpath \vert_{\hV_{\hH}} \middle\bracevert \hu\seq\hpath \in \Conn_{\hH} ~\tand~ \seq\hpath \in \hV_\hH^* \right\}.
\]
Note that the coloring of any concrete path is a trace.

\begin{proposition}\label{lem:repl-bound}
	Let $\hH \in \hHH$ and $\fml{\hJ}$ be a hyperedge assignment to $\hH$. Then, let $\he,\he_1,\he_2 \in \hE_{\hH}$ such that $\he_1 \ne \he_2$, and $\hu \in \hW_{\hH}$, $\hv \in \hW_{\fml\hJ(\he)}$, $\hv_1 \in \hW_{\fml\hJ(\he_1)}$, $\hv_2 \in \hW_{\fml\hJ(\he_2)}$. Then, 
	\begin{enumerate}
		\item $(\hu \rightarrow_{\repl{\hH}{\fml{\hJ}}} \hv) \iff \exists i.\; \hu = \hd_{\hH}(\he)(i) \land (\hI_{\fml\hJ(\he)}(i) \rightarrow_{\fml\hJ(\he)} \hv)$,
		\item $(\hv \rightarrow_{\repl{\hH}{\fml{\hJ}}} \hu) \iff \exists i.\; \hu = \hd_{\hH}(\he)(i) \land (\hv \rightarrow_{\fml\hJ(\he)} \hI_{\fml\hJ(\he)}(i))$,
		\item $\hv_1 \not\rightarrow_{\repl{\hH}{\fml{\hJ}}} \hv_2$.
	\end{enumerate}	
\end{proposition}

\begin{proof}
	Let $\hH, \fml{\hJ}$ as described. Recall that, following the definition of hyperedge replacement, for all $\hu,\hv$,
	$$
	\hu \rightarrow_{\repl{\hH}{\fml{\hJ}}} \hv
		\iff (\hu,\hv) \in \left(\hA_{\hH} \uplus \left(\biguplus_{\he \in \hE_{\hH}} f_\he(\hA_{\fml\hJ(\he)})\right)\right).
	$$
	We begin by noting that each part of this union contains pairs that refer to nodes in different \HGs: $\hA_{\hH}$ contains nodes that are exclusively in $\hH$, while each $Rf_\he(\hA_{\fml\hJ(\he)})$ contain nodes that can be both in $\hH$ and $\fml\hJ(\he)$. We can immediately conclude that there is no pair that can contain both $\hv_1 \in \hW_{\hJ_{\he_1}}$ and $\hv_2 \in \hW_{\hJ_{\he_2}}$ for $\he_1 \ne \he_2$. If, instead, we consider $\hu \in \hW_{\hH}$ and $\hv \in \hW_{\fml\hJ(\he)}$, we note that the only component of the union that may contain this pair is $f_\he(\hA_{\fml\hJ(\he)})$. This means that
	\begin{align*}
		(\hu \rightarrow_{\repl{\hH}{\fml{\hJ}}} \hv) 
		\iff& (\hu,\hv) \in f_\he(\hA_{\fml\hJ(\he)})
		\\\iff& \exists \hu',\hv'.\; (\hu' \hA_{\fml\hJ(\he)} \hv') \land \hu = f_\he(\hu') \land \hv = f_\he(\hv').
	\intertext{Note that, as $\hu \in \hW_{\hH}$, $\hu'$ must be equal to $\hI_{\fml\hJ(\he)}(i)$ for some $i$, while, as $\hv \in \hW_{\fml\hJ(\he)}$, $\hv' = \hv$. Therefore,}
		\iff&\exists i.\; (\hI_{\fml\hJ(\he)}(i) \rightarrow_{\fml\hJ(\he)} \hv) \land \hu = \hd_{\hH}(\he)(i)
	\end{align*}
	as required. The proof for the $\hv \rightarrow_{\repl{\hH}{\fml{\hJ}}} \hu$ case is equivalent.
\end{proof}

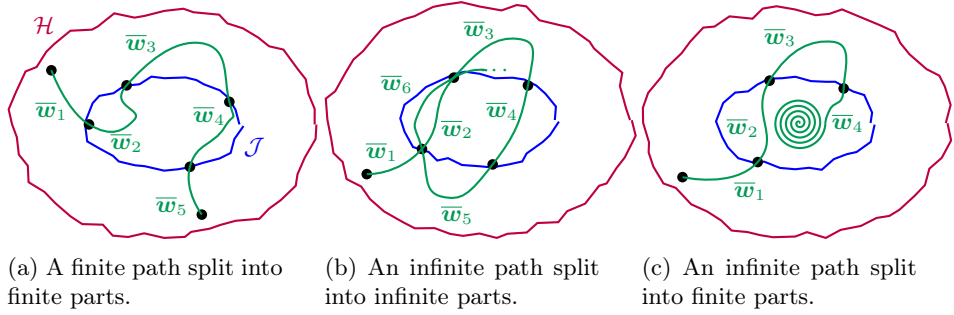
\begin{figure}[t]
\centering
\def\a{1} 
\def\b{0.6}
\begin{subfigure}[b]{0.3\linewidth}
\centering
\scalebox{1}{\begin{tikzpicture}
\draw[thick, purple, decorate, decoration={random steps, segment length=6pt, amplitude=2pt}]
  (0,0) ellipse (2 and 1.5);

\draw[thick, blue, decorate, decoration={random steps, segment length=5pt, amplitude=1.5pt}]
  (0,0) ellipse (1 and 0.6);

\foreach \name/\angle in {P1/180, P2/120, P3/30, P4/290} {
  \coordinate (\name) at ({\a*cos(\angle)}, {\b*sin(\angle)});
  \fill (\name) circle (2pt);
}

\coordinate (X) at (\a*-1.5,\b*1.2);
\coordinate (Y) at (\a*0.5,\b*-2);
\foreach \p in {X, Y} {
  \fill (\p) circle (2pt);
};

\draw[thick, ForestGreen] plot[smooth, tension=1] coordinates { (X) (P1) (\a*-0.4,\b*0.1) (P2) (\a*0.5,\b*1.8) (P3) (\a*0.8,\b*-0.07) (P4) (Y) };

\node[ForestGreen] at (\a*-1.5,\b*0.3) {$\seq\hw_1$};
\node[ForestGreen] at (\a*-0.5,\b*-0.4) {$\seq\hw_2$};
\node[ForestGreen] at (\a*-0.3,\b*1.8) {$\seq\hw_3$};
\node[ForestGreen] at (\a*0.6,\b*0.2) {$\seq\hw_4$};
\node[ForestGreen] at (\a*0.1,\b*-1.8) {$\seq\hw_5$};

\node[purple] at (\a*-1.6,\b*2.2) {$\hH$};
\node[blue] at (\a*1.2,\b*-0.5) {$\hJ$};

\end{tikzpicture}}
\caption{A finite path split into finite parts.}
\label{fig:path-decomposition}
\end{subfigure}
\hfill
\begin{subfigure}[b]{0.3\linewidth}
\centering
\scalebox{1}{\begin{tikzpicture}
\draw[thick, purple, decorate, decoration={random steps, segment length=6pt, amplitude=2pt}]
  (0,0) ellipse (2 and 1.5);

\draw[thick, blue, decorate, decoration={random steps, segment length=5pt, amplitude=1.5pt}]
  (0,0) ellipse (1 and 0.6);
  
\foreach \name/\angle in {P1/220, P2/110, P3/50, P4/280} {
  \coordinate (\name) at ({\a*cos(\angle)}, {\b*sin(\angle)});
  \fill (\name) circle (2pt);
}

\coordinate (X) at (\a*-1.5,\b*-1.2);
\foreach \p in {X} {
  \fill (\p) circle (2pt);
};

\draw[thick, ForestGreen] plot[smooth, tension=1] coordinates { (X)
	(P1) (P2) (\a*0.5,\b*1.8) (P3) (P4)  (\a*-0.5,\b*-1.7) 
	(P1) (\a*-0.8, \b*0.2) (P2) (\a*0.1, \b*1.1) }; 

\node[ForestGreen] at (\a*-1.3,\b*-0.7) {$\seq\hw_1$};
\node[ForestGreen] at (\a*-0.3,\b*-0.2) {$\seq\hw_2$};
\node[ForestGreen] at (\a*-0.0,\b*2) {$\seq\hw_3$};
\node[ForestGreen] at (\a*0.3,\b*0.2) {$\seq\hw_4$};
\node[ForestGreen] at (\a*-0.3,\b*-2.1) {$\seq\hw_5$};
\node[ForestGreen] at (\a*-1.1,\b*0.8) {$\seq\hw_6$};
\node[ForestGreen] at (\a*0.22,\b*1.09) {$\cdots$};

\end{tikzpicture}}
\caption{An infinite path split into infinite parts.}
\label{fig:path-decomposition}
\end{subfigure}
\hfill
\begin{subfigure}[b]{0.3\linewidth}
\centering
\scalebox{1}{\begin{tikzpicture}
\newcommand\spiral{}
\def\spiral[#1](#2)(#3:#4:#5){
\pgfmathsetmacro{\domain}{pi*#3/180+#4*2*pi}
\draw [#1,shift={(#2)}, domain=0:\domain,variable=\t,smooth,samples=int(\domain/0.08)] plot ({\t r}: {#5*\t/\domain})
}

\draw[thick, purple, decorate, decoration={random steps, segment length=6pt, amplitude=2pt}]
  (0,0) ellipse (2 and 1.5);

\draw[thick, blue, decorate, decoration={random steps, segment length=5pt, amplitude=1.5pt}]
  (0,0) ellipse (1 and 0.6);
  
\foreach \name/\angle in {P1/240, P2/110, P3/50} {
  \coordinate (\name) at ({\a*cos(\angle)}, {\b*sin(\angle)});
  \fill (\name) circle (2pt);
}

\coordinate (X) at (\a*-1.5,\b*-1.2);
\foreach \p in {X} {
  \fill (\p) circle (2pt);
};

\draw[thick, ForestGreen] plot[smooth, tension=1] coordinates { (X)
	(P1) (P2) (\a*0.5,\b*1.8) (P3) (\a*0.46,\b*0.3) (\a*0.40,\b*0.0) };
	
\spiral[thick, ForestGreen](\a*0.05,\b*0.0)(0:5:0.35);

\node[ForestGreen] at (\a*-0.6,\b*-1.5) {$\seq\hw_1$};
\node[ForestGreen] at (\a*-0.7,\b*-0.1) {$\seq\hw_2$};
\node[ForestGreen] at (\a*-0.2,\b*1.9) {$\seq\hw_3$};
\node[ForestGreen] at (\a*0.7,\b*-0.0) {$\seq\hw_4$};
  
\end{tikzpicture}}

\caption{An infinite path split into finite parts.}
\label{fig:path-decomposition}
\end{subfigure}
\caption{Paths on a composed hypergraph split into components.}
\label{fig:path-decomposition}
\end{figure}

\begin{proposition}[Path decomposition]\label{lem:path-decomposition}
Let $\hH \in \hHH$, $\fml{\hJ}$ be a hyperedge assignments and let $\seq\hw \in \CPaths^*_{\repl{\hH}{\fml{\hJ}}}(\hu,\hv)$ for $\hu,\hv \in \hW_{\repl{\hH}{\fml{\hJ}}}$. Then, $\seq\hw$ can be partitioned into a sequence of paths $\seq\hw_1 \cdots \seq\hw_n$, such that for all $i$, there exists $\hK \in (\biguplus_\he \fml\hJ(\he)) \uplus \{\hH\}$ such that 
$$
	\exists \hu',\hv' \in \left(\{\hu,\hv\} \uplus \biguplus_{\he} \hI_{\fml\hJ(\he)} \uplus \hd_{\hH}(\he)\right) \cap \hW_{\hK}.\; \seq\hw_i \in  \CPaths^*_{\hK}(\hu',\hv').
$$  
Similarly, let $\seq\hw \in \CPaths^\omega_{\repl{\hH}{\fml{\hJ}}}(\hu)$. Then, $\seq\hw$ can be partitioned into either an infinite sequence of finite paths $\seq\hw_1 \cdots$ with every $\seq\hw_i$ behaving as before, or into a finite sequence of paths $\seq\hw_1 \cdots \seq\hw_{n-1}\seq\hw_n$ where $\seq\hw_n$ is an infinite path for which additionally there exists $\hK \in (\biguplus_\he \fml\hJ(\he)) \uplus \{\hH\}$ such that 
$$
	\exists \hu' \in \left(\{\hu\} \uplus \biguplus_{\he} \hI_{\fml\hJ(\he)} \uplus \hd_{\hH}(\he)\right) \cap \hW_{\hK}.\; \seq\hw_n \in \CPaths^\omega_{\hK}(\hu').
$$
This behaviour is pictured in \Cref{fig:path-decomposition}.
\end{proposition}

\begin{proof}
	Let $\hH,\fml{\hJ},\hu,\hv$ and $\seq{\hw}$ as described. We split $\seq{\hw}$ in a sequence of non-empty paths $\seq{\hw}_1\seq{\hw}_2\cdots\seq{\hw}_n$ such that each path is composed only of states in a single \HG, splitting only when the nodes belong to a different hypergraph. Consider the path $\seq{\hw}_k$ for arbitrary $k$. We split our proof into multiple cases, depending on whether $\seq{\hw}_k \in \hw_{\hH}^*$ or $\seq{\hw}_k \in \hw_{\fml\hJ(\he)}^*$ (for some $\he$), and on whether $1=k<n$, $1<k<n$, $1<k=n$ or $1=k=n$. We prove in detail some of the more illustrative cases:
	\begin{itemize}
		\item Case $\hH$, $1<k<n$ : as the sequence $\seq{\hw}_k$ is comprised of nodes in $\hH$, the sequence $\seq{\hw}_{k-1}$ must be composed of nodes in $\hJ_{\he_1}$, while $\seq{\hw}_{k+1}$ must be composed of nodes in $\hJ_{\he_2}$ (where $\he_1$ and $\he_2$ may be the same), as otherwise they would be part of $\seq\hw_k$. As $\seq\hw_1 \dots \seq\hw_{k-1} \seq\hw_k \seq\hw_{k+1} \dots \seq{\hw}_n$ is a path, it must hold that $\last(\seq\hw_{k-1}) \rightarrow_{\repl{\hH}{\fml{\hJ}}} \first(\seq\hw_k)$ and $\last(\seq\hw_k) \rightarrow_{\repl{\hH}{\fml{\hJ}}} \first(\seq\hw_{k+1})$. By \Cref{lem:repl-bound}, this can only be true if
			$$
				\exists i.\; \last(\seq\hw_k) = \hd_{\hH}(\he_1)(i) 
			\quad\text{ and }\quad
				\exists i.\; \first(\seq\hw_k) = \hd_{\hH}(\he_2)(i) 
			$$
			therefore,
			$$
				\exists \hu',\hv' \in \hd_{\hH}(\he_1) \cup \hd_{\hH}(\he_2).\; \seq\hw_k \in \CPaths_{\hH}^*(\hu',\hv').
			$$
		\item Case $\fml\hJ(\he)$, $1=k<n$: firstly, note that since $\seq\hw_1 \dots \seq\hw_{k-1} \seq\hw_k \seq\hw_{k+1} \dots \seq\hw_n \in \CPaths^*(\hu,\hv)$, it follows that
			$$
				\seq\hw_1 \in \CPaths^*(\hu, \last(\seq\hw_1)).
			$$ 
			Note that since the sequence $\seq{\hw}_1$ is composed of nodes in $\fml\hJ(\he)$, then $\seq{\hw}_2$ must be composed of nodes in $\hH$, as by \Cref{lem:repl-bound} it is not possible for nodes in $\fml\hJ(\he)$ and $\hJ_{\he'}$ for $\he'\ne\he$ to be directly connected. This means that, by \Cref{lem:repl-bound},
			$$
				\exists i.\; \first(\seq\hw_2) = \hd_{\hH}(\he)(i) \land (\last(\seq\hw_{1}) \rightarrow_{\fml\hJ(\he)} \hI_{\fml\hJ(\he)}(i)),
			$$
			Therefore, as $\last(\seq\hw_1) \in \hw_{\fml\hJ(\he)}$ and $\last(\seq\hw_{1}) \rightarrow_{\fml\hJ(\he)} \hI_{\fml\hJ(\he)}(i)$,
			$$
				\seq\hw_1 \in \CPaths_{\hH}^*(\hu,\hI_{\fml\hJ(\he)}(i)),
			$$
			and thus
			$$
				\exists \hv' \in \hI_{\fml\hJ(\he)}.\; \seq\hw_1 \in \CPaths_{\hH}^*(\hu,\hv')).
			$$
		\item The other cases follow the same idea: if the path is at the beginning, then it will start with $\hu$, if it last, it will end with $\hv$. Depending on whether it is composed of nodes in $\hH$ or $\fml\hJ(\he)$, it will force its neighboring paths to be, respectively, in $\hJ_{\he'}$ or $\hH$, forcing it to start and end with either abstract nodes or attached nodes.
	\end{itemize}
	Next, let $\seq\hw$ be an infinite path as described in the theorem. We split it into a sequence of non-empty paths as done before. This could result in either a sequence $\seq\hw_1\seq\hw_2\cdots$, or $\seq\hw_1\cdots\seq\hw_{n-1}\seq\hw_n$ (with $\seq\hw_n$ being infinite). The proof for the first case follows the previous idea, while in the second case we additionally need to discuss the sequence $\seq\hw_n$. Similarly to before, we distinguish depending on whether (1) $\seq{\hw}_n \in \hH$ or $\seq{\hw}_n \in \fml\hJ(\he)$ for some $\he$, and (2) $1=k=n$, $1<k=n$.
	\begin{itemize}
		\item Case $\hH$, $1<k=n$: the sequence $\seq{\hw}_n$ is composed of nodes in $\hH$, forcing the sequence $\seq{\hw}_{n-1}$ to be composed of nodes in $\fml\hJ(\he)$ for some $\he$. By \Cref{lem:repl-bound}, we have that
			$$
				\exists i.\; \first(\seq\hw_n) = \hd_{\hH}(\he)(i) \land (\last(\seq\hw_{n-1}) \rightarrow_{\fml\hJ(\he)}  \hI_{\fml\hJ(\he)}(i)),
			$$
			and therefore,
			$$
				\exists\hu' \in \hd_{\hH}(\he).\; \seq\hw_n \in \CPaths^\omega(\hu').
			$$
		\item The other cases follow the same structure.
	\end{itemize}
\end{proof}
	
\begin{proposition} \label{lem:string-meq-join}
	$\bMM$-equivalence of finite and infinite words is closed under finite and infinite concatenation. That is, for $(\seq\bs_i, \seq\bs_i' \in \bS^*)_i$, $\seq{\rho}, \seq{\rho}' \in \bS^\omega$ such that $\forall i.\; \seq\bs_i \equiv_\bMM \seq\bs_i'$ and $\seq{\rho} \equiv_\bMM \seq{\rho}'$ we have that
	$$
		\seq\bs_1\seq\bs_2\cdots\seq\bs_n \equiv_\bMM \seq\bs_1'\seq\bs_2'\cdots\seq\bs_n', \quad 
		\seq\bs_1\seq\bs_2\cdots \equiv_\bMM \seq\bs_1'\seq\bs_2'\cdots \quad \text{and} \quad
		\seq\bs_1\seq{\rho} \equiv_\bMM \seq\bs_1'\seq{\rho}'.
	$$
\end{proposition}
\begin{proof}
\allowdisplaybreaks
For convenience, in the following proof we will split runs $\bq_0\bs_0\bq_1\bs_1\cdots\bs_n\bq_n \in \Runs^*_\bMM$ into $(\bq_0, \bs_0\bs_1\cdots\bs_n, \bq_1\bq_2\cdots\bq_n) \in \Runs^*_\bMM$. We proceed by proving the three claims separately.

\textit{Claim 1.} We begin by proving that $\seq\bs_1\seq\bs_2 \equiv_\bMM \seq\bs_1'\seq\bs_2'$. That is, for arbitrary $\bp,\bq,\boolb$,
\begin{align*}
	&(\bp,\seq\bs_1\seq\bs_2,\bq,\boolb) \in \bD^* 
	\\
	\iff&\exists \seq\br.\; (\bp, \seq\bs_1\seq\bs_2,\seq\br) \in \Runs^*_{\bMM} \land \bq = \last(\bp\seq\br) \land b = (\bp\seq\br \vert_{\bF_{\bMM}} \ne \epsilon)
	\intertext{We split the sequence of states $\seq\br$ into two sequences $\seq\br_1$ and $\seq\br_2$, such that $|\seq\br_1| = \seq\bs_1$ and $|\seq\br_2| = \seq\bs_2$.}
	\iff&\exists \seq\br_1\seq\br_2.\; (\bp, \seq\bs_1\seq\bs_2,\seq\br_1\seq\br_2) \in \Runs^*_{\bMM} \land \bq = \last(\bp\seq\br_1\seq\br_2) \land \boolb = (\bp\seq\br_1\seq\br_2 \vert_{\bF_{\bMM}} \ne \epsilon)
	\intertext{We split the runs into two, passing via the intermediate state $\bp'$.}
	\iff&\exists \bp', \seq\br_1\seq\br_2.\;  (\bp, \seq\bs_1,\seq\br_1)  \in \Runs^*_{\bMM} \land (\bp', \seq\bs_2,\seq\br_2) \in \Runs^*_{\bMM} \land \bp' = \last(\bp\seq\br_1)
		\\&\qquad \land \bq = \last(\bp'\seq\br_2) \land 	\boolb = (\bp\seq\br_1 \vert_{\bF_{\bMM}} \ne \epsilon \lor \bp'\seq\br_2 \vert_{\bF_{\bMM}} \ne \epsilon)
	\intertext{We concrete split $\boolb$ into two $\boolb_1$ and $\boolb_2$}
	\iff&\exists \bp', \boolb_1, \boolb_2, \seq\br_1\seq\br_2.\;  (\bp, \seq\bs_1,\seq\br_1)  \in \Runs^*_{\bMM} \land (\bp', \seq\bs_2,\seq\br_2) \in \Runs^*_{\bMM} \land \bp' = \last(\bp\seq\br_1)
		\\&\qquad \land \bq = \last(\bq'\seq\br_2) \land \boolb = (\boolb_1 \lor \boolb_2) \land \boolb_1 = (\bp\seq\br_1 \vert_{\bF_{\bMM}} \ne \epsilon) \land \boolb_2 = (\bp'\seq\br_2 \vert_{\bF_{\bMM} \ne \epsilon})
	\intertext{We wrap everything into $\bD^*$ again}
	\iff&\exists \bp',\boolb_1,\boolb_2.\; (\bp,\seq\bs_1,\bp',\boolb_1) \in \bD^* \land (\bp',\seq\bs_2,\bq,\boolb_2) \in \bD^* \land (\boolb = \boolb_1 \lor \boolb_2)
	\intertext{By $\bMM$-equivalence of $\seq\bs_1$ with $\seq\bs_1'$ and equivalence of $\seq\bs_2$ with $\seq\bs_2'$, we have}
	\iff&\exists \bp',\boolb_1,\boolb_2.\; (\bp,\seq\bs_1',\bp',\boolb_1) \in \bD^* \land (\bp',\seq\bs_2',\bq,\boolb_2) \in \bD^* \land (\boolb = \boolb_1 \lor \boolb_2)
	\intertext{We run all the steps backwards to go back to our starting formula, but on $\seq\bs_1'$ and $\seq\bs_2'$}
	\iff&(\bp,\seq\bs_1'\seq\bs_2',\bq,\boolb) \in \bD^* 
\end{align*} 
Then, by induction, $\seq\bs_1\seq\bs_2\cdots\seq\bs_n \equiv_\bMM \seq\bs_1'\seq\bs_2'\cdots\seq\bs_n'$. 

\textit{Claim 2.} Similarly, we prove that $\seq\bs\seq{\rho} \equiv_\bMM \seq\bs_1'\seq{\rho}'$. That is, for arbitrary $p$,
\begin{align*}
	&(\bp,\seq\bs\seq\rho,,\boolT) \in \bD^\omega 
	\\\iff& \exists \seq\br.\; (\bp, \seq\bs_1\seq\rho, \seq\br) \in \Runs^\omega_{\bMM} \land \Inf(\seq\br) \cap \bF_{\bMM} \ne \emptyset
	\\\iff& \exists \seq\br_1\seq\br_2.\; (\bp, \seq\bs\seq\rho, \seq\br_1\seq\br_2) \in \Runs^\omega_{\bMM} \land \Inf(\seq\br_1\seq\br_2) \cap \bF_{\bMM} \ne \emptyset
	\\\iff& \exists \bq,\seq\br_1\seq\br_2.\; (\bp, \seq\bs, \seq\br_1) \in \Runs^*_{\bMM} \land (\bq, \seq\rho, \seq\br_2) \in \Runs^\omega_{\bMM} \land \bq = \last(\bp\seq\br_1) 
		\\&\qquad \land \Inf(\seq\br_2) \cap \bF_{\bMM} \ne \emptyset
	\\\iff&\exists \bq,b.\; (\bp,\seq\bs,\bq,b) \in \bD^* \land (\bq,\seq\rho,\boolT) \in \bD^\omega 
	\\\iff&\exists \bq,b.\; (\bp,\seq\bs',\bq,b) \in \bD^* \land (\bq,\seq\rho',\boolT) \in \bD^\omega 
	\\\iff& \cdots
	\\\iff&(\bp,\seq\bs'\seq\rho',\boolT) \in \bD^\omega.
\end{align*}

\textit{Claim 3.} Lastly, we show $\seq\bs_1\seq\bs_2\cdots \equiv_\bMM \seq\bs_1'\seq\bs_2'\cdots$. That is, for arbitrary $p$
\begin{align*}
	&(\bp, \seq\bs_1\seq\bs_2\cdots,\boolT) \in \bD^\omega 
	\\\iff& \exists \seq\br.\; (\bp, \seq\bs_1 \cdots, \seq\br) \in \Runs^\omega_{\bMM} \land \Inf(\seq\br) \cap \bF_{\bMM} \ne \emptyset
	\\\iff& \exists \seq\br_1\seq\br_2\cdots.\; (\bp, \seq\bs_1 \seq\bs_2 \cdots, \seq\br_1 \seq\br_2 \cdots) \in \Runs^\omega_{\bMM} \land \Inf(\seq\br_1\seq\br_2\cdots) \cap \bF_{\bMM} \ne \emptyset
	\\\iff& \exists \seq\br_1\seq\br_2\cdots.\; (\bp, \seq\bs_1 \seq\bs_2 \cdots, \seq\br_1 \seq\br_2 \cdots) \in \Runs^\omega_{\bMM} \land \forall i.\; \exists k \ge i.\; \seq\br_k \vert_{\bF_{\bMM}} \ne \epsilon
	\\\iff& \exists \seq\boolb, \seq\br_1\seq\br_2\cdots.\; (\bp, \seq\bs_1 \seq\bs_2\cdots, \seq\br_1 \seq\br_2\cdots) \in \Runs^\omega_{\bMM} \land \forall i.\; \boolb_i = (\seq\br_i \vert_{\bF_{\bMM}} \ne \epsilon) 
		\\&\qquad \land \exists k \ge i.\; \boolb_i = \boolT
	\\\iff& \exists \seq\boolb, \seq\br_1\seq\br_2\cdots, \seq\bq.\; \bp = \first(\seq\bq) \land \forall i.\; (\bq_i, \seq\bs_i, \seq\br_i) \in \Runs^*_{\bMM} \land \bq_{i+1} = \last(\bq_i\seq\br_i) 
		\\&\qquad \land \boolb_i = (\seq\br_i \vert_{\bF_{\bMM}} \ne \epsilon) \land \exists k \ge i.\; \boolb_i = \boolT
	\\\iff& \exists \seq\boolb, \seq\bq.\; \bp = \first(\seq\bq) \land \forall i.\; \exists \seq\br.\; (\bq_i, \seq\bs_i, \seq\br) \in \Runs^*_{\bMM} \land \bq_{i+1} = \last(\bq_i\seq\br) 
		\\&\qquad \land \boolb_i = (\seq\br \vert_{\bF_{\bMM}} \ne \epsilon) \land \exists k \ge i.\; \boolb_i = \boolT
	\\\iff& \exists \seq\boolb, \seq\bq.\; \bp = \first(\seq\bq) \land \forall i.\; (\bq_i, \seq\bs_i, \bq_{i+1}, \boolb_i) \in \bD^*_{\bMM} \land \exists k \ge i.\; \boolb_i = \boolT
	\\\iff& \exists \seq\boolb, \seq\bq.\; \bp = \first(\seq\bq) \land \forall i.\; (\bq_i, \seq\bs_i', \bq_{i+1}, \boolb_i) \in \bD^*_{\bMM} \land \exists k \ge i.\; \boolb_i = \boolT
	\\\iff& \cdots
	\\\iff& (\bp, \seq\bs_1'\seq\bs_2'\cdots,\boolT) \in \bD^\omega 
\end{align*}
\end{proof}

\repeattheorem{lem:meq-cong}
\begin{proof}
Let $(\eta,\mu) \vdash \hUU \equiv_\bMM \hUU'$ and $(\eta_\he',\mu_\he') \vdash \fml\hVV(\he) \equiv_\bMM \hVV'(\mu(\he))$ for $\he \in \hE_\hH$. We want to show that $(\eta'', \mu'') \vdash \repl{\hUU}{\fml\hVV} \equiv_\bMM \repl{\hUU'}{\fml\hVV'}$ for $\eta'' = \biguplus_\he \eta'_\he \uplus \eta$ and $\mu'' = \biguplus_\he \mu'_\he$, that is, we want to show that  
$$
	\forall \seq\hw \in \CPaths^*_{\repl\hUU{\fml\hVV}}(x,y).\; \exists\seq\hw' \in \CPaths^*_{\repl{\hUU'}{\fml\hVV'}}(\nu(x),\nu(y)).\; \hC(\seq\hw) \equiv_\bMM \hC(\seq\hw')
$$ 
as well as
$$
	\forall \seq\hw \in \CPaths^\omega_{\repl\hUU{\fml\hVV}}(x).\; \exists\seq\hw' \in \CPaths^\omega_{\repl{\hUU'}{\fml\hVV'}}(\nu(x)).\; \hC(\seq\hw) \equiv_\bMM \hC(\seq\hw').
$$ 
Let $\seq\hw \in \CPaths^*_{\repl{\hUU}{\fml\hVV}}(x,y)$. From \Cref{lem:path-decomposition}, we can split this path into a sequence of paths $\seq\hw_1\seq\hw_2\cdots\seq\hw_n$ such that
$$
	\exists \hu',\hv' \in \hW.\; \seq\hw_i \in  \CPaths^*_{\hK}(x',y').
$$  
for $\hK \in \bigcup_\he \hVV(\he) \cup \{\hUU\}$ and $\hW = \left(\{\hu,\hv\} \uplus \biguplus_{\he} \hI_{\fml\hJ(\he)} \uplus \hd_{\hH}(\he)\right) \cap \hW_{\hK}$. Assume that $\hK = \hH$, then $\hW \subseteq X_\hVV$. Thus, as $\hUU$ and $\hUU'$ are $\bMM$-equivalent, there must exist a path $\seq\hw_i' \in \CPaths_{\hUU'}(\nu(x),\nu(y))$ such that $\hC(\seq\hw_i) \equiv_\bMM \hC(\seq\hw_i')$. Applying the same reasoning to all subpaths and joining them yields a new path $\seq\hw_i'$ that is $\bMM$-equivalent to the original path $\seq\hw$ (By \Cref{lem:string-meq-join}). A similar line of reasoning can be applied to infinite paths.
\end{proof}

\subsection{Proof of \Cref{lem:meq-sem}}

\repeatlemma{lem:meq-sem}

\begin{proof}
	Let $\hVV = (\hH,U)$ and $\hVV' = (\hH',U')$. Let $(\eta,\mu) \vdash \hVV \equiv_\bMM \hVV'$. Then, for $x \in X_\hVV$,
	\begin{align*}
		(\hVV, x) \models \bMM
		   &\iff \forall \seq{\hc} \in \Paths_{\hVV}^\omega(x). \; \seq{\hc} \models \bMM
		   \intertext{Then, consider an arbitrary trace $\seq{\hc}' \in \Paths_{\hVV'}^\omega(\nu(x))$. By definition of $\bMM$-equivalence, it follows that there exists a trace $\seq{\hc} \in \Paths_{\hVV'}^\omega(x)$ $\bMM$-equivalent to $\seq{\hc}'$, and therefore such that $\seq{\hc}' \models \bMM$. The same argument can be applied in the reverse direction, thus}
		   &\iff \forall \seq{\hc}' \in \Paths_{\hVV'}^\omega(\nu(x)). \; \seq{\hc}' \models \bMM
		\\ &\iff (\hVV', \nu(x)) \models \bMM.
	\end{align*}
\end{proof}
\subsection{Proof of \Cref{lem:meq-mbhv}}

\repeattheorem{lem:meq-mbhv}

\begin{proof}
We note that isomorphism of $\bMM$-behaviours can be rephrased in terms of equality of the edge relations $\rightsquigarrow^*$ and $\rightsquigarrow^\omega$, under the mapping $\nu$. Then, we note that 
\begin{align*}	
	\{(x',y',\{(\bp,\bq,\boolb) \mid \bD^*(\bp,\seq\hc,\bq,\boolb)\}) &\mid \seq\hc \in \Paths^*_{\hVV'}(x',y') \} \\
=
	\{(\nu(x),\nu(y),\{(\bp,\bq,\boolb) \mid \bD^*(\bp,\seq\hc,\bq,\boolb)\}) &\mid \seq\hc \in \Paths^*_{\hVV}(x,y) \}),
\end{align*}
iff
$$
	\Traces_\hVV^*(x,y) \equiv_\bMM \Traces^*_{\hVV'}(\nu(x),\nu(y)).
$$
Similarly, 
\begin{align*}
	\{(x',\{ \bp \mid \bD^*(\bp, \seq\hc, \boolT) \}) &\mid \seq\hc \in \Paths^\omega_{\hVV'}(x') \} \\
=
	\{(\nu(x), \{ \bp \mid \bD^*(\bp,\seq\hc, \boolT) \}) &\mid \seq\hc \in \Paths^\omega_\hVV(x) \}).
\end{align*}
iff 
$$
\Traces_\hVV^\omega(x) \equiv_\bMM \Traces_{\hVV'}^\omega(\nu(x)).
$$
Our claim follows.
\end{proof}
\subsection{Proof of \Cref{thm:amt-meq}}

\repeattheorem{thm:amt-meq}

\begin{proof}
\Cref{lem:meq-mbhv} tells us that $\bMM$-equivalence can be checked via isomorphism of $\bMM$-behaviours. From this follows that the number of equivalence classes on no nodes and no hyperedges only depends on the edges that are present between the $n$ abstract nodes, as well as the node $\star$. This leads us to
\begin{align*}	
	|\eqclass{(\lLL(\gGG), \emptyset)}_{\equiv\bMM}|
	&= \bigO\left(2^{m \times m \times 2^{|\bQ \times \bQ \times \Bool|} + m \times 1 \times 2^{|\bQ|}}\right) \\
	&= \bigO\left(2^{m^2 \times 2^{|\bQ|^2}}\right)
\end{align*}
Similarly, equivalence on hypergraphs with no abstract nodes and only one hyperedge with $n$ attached nodes depends only on the edges that are present between the $n$ attached nodes, as well as the node $star$, leading us to the same calculation.

We then use this result to calculate the size of $\gGG_{\equiv\bMM}$. Note that the size of the grammar only relevantly depends on the number of rules, as any nonterminal that does not show up in a rule can be removed.

\begin{align*}
	|\gGG_\bMM|
	&= \bigO(|\gP_{\gGG_{\bMM}}|) \\
	&= \bigO\left(|\gN_{\gGG_{\bMM}}| \times |\gP_\gGG| \times |\gN_{\gGG_{\bMM}}|^{|\hE|} \right) \\
	&= \bigO\left(|\gP_\gGG| \times |\gN_{\gGG_{\bMM}}|^{|\hE| + 1} \right) \\
	&= \bigO\left(|\gP_\gGG| \times \left(|\gN_\gGG| \times |\eqclass{(\lLL(\gGG), \emptyset)}| \times |\eqclass{(\gCC(\gGG), \emptyset)}|\right)^{|\hE| + 1} \right) \\
	&= \bigO\left(|\gP_\gGG| \times \left(|\gN_\gGG| \times 2^{|\hI|^2 \times 2^{|\bQ|^2}} \times 2^{|\hI|^2 \times 2^{|\bQ|^2}}\right)^{|\hE|} \right)  \\
	&= \bigO\left(|\gP_\gGG| \times \left(|\gN_\gGG| \times 2^{\left( 2 \times \left(|\hI|^2 \times 2^{|\bQ|^2}\right)\right)}\right)^{|\hE|} \right)  \\
	&= \bigO\left(|\gP_\gGG| \times |\gN_\gGG|^{|\hE|} \times 2^{\left( |\hE| \times |\hI|^2 \times 2^{|\bQ|^2}\right) } \right).
\end{align*}
\end{proof}

\end{document}